\documentclass[12pt]{article}
\usepackage[T1]{fontenc}
\usepackage[utf8]{inputenc}
\usepackage{amsmath, amsthm, amssymb,amsfonts}
\usepackage{marginnote}
\usepackage{fullpage}
\usepackage{color}
\usepackage[dvipsnames]{xcolor}
\usepackage{graphicx}
\usepackage{cite}
\usepackage{hyperref}
\usepackage{bbm}
\usepackage{bbold}

\graphicspath{{./Figures/}}
\newcommand{\dhat}{\widehat{d}}
\newcommand{\dtilde}{\tilde{d}}
\newcommand{\eps}{\varepsilon}
\newcommand{\epshat}{\widehat{\eps}}
\newcommand{\e}{\mathrm{e}}
\newcommand{\R}{\mathbb{R}}

\newcommand{\polylog}{\mathrm{polylog}}
\newcommand{\sym}{\mathrm{Sym}}
\newcommand{\vol}{\mathrm{vol}}

\newtheorem{theorem}{Theorem}[section]
\newtheorem{definition}[theorem]{Definition}
\newtheorem{lemma}[theorem]{Lemma}
\newtheorem{remark}[theorem]{Remark}
\newtheorem{corollary}[theorem]{Corollary}
\newtheorem{proposition}[theorem]{Proposition}

\newcommand{\SRn} {\ensuremath{[0,\sqrt{n}]^2}}
\newcommand{\Prob}[1]{\mathbb{P} \left( #1 \right)}

\newcommand{\whp}{w.h.p.}
\newcommand{\RG} {\ensuremath{\mathcal G(n,r)}}

\newcommand{\RGp} {\ensuremath{\mathcal G(n,r,p)}}
\newcommand{\remove}[1]{}

\DeclareMathOperator{\Exp}{\mathbb{E}}
\newcommand{\ExpCond}[2]{\mathbb{E}\left(#1 \mid #2\right)}
\newcommand{\jdc}[1]{{\color{red}#1}}

\begin{document}
\date{}
\title{Reconstruction of Random Geometric Graphs: Breaking the $\Omega (r)$ distortion barrier}

\author{Varsha Dani\thanks{Dept. of Computer Science, Rochester Institute of Technology, Rochester, NY}
\and Josep D\'{i}az\thanks{Dept. Comput. Science, UPC, Barcelona. Partially supported  by PID-2020-112581GB-C21 (MOTION)}
\and Thomas P. Hayes\thanks{Dept. Computer Science, U. New Mexico, Albuquerque, NM.}
\and Cristopher Moore\thanks{Santa Fe Institute, Santa Fe, NM. Partially supported by NSF grant IIS-1838251}}

\maketitle

\thispagestyle{empty}
\begin{abstract}
Embedding graphs in a geographical or latent space, i.e.\ inferring locations for vertices in Euclidean space or on a smooth manifold or submanifold, is a common task in network analysis, statistical inference, and graph visualization. We consider the classic model of random geometric graphs where $n$ points are scattered uniformly in a square of area $n$, and two points have an edge between them if and only if their Euclidean distance is less than $r$. The reconstruction problem then consists of inferring the vertex positions, up to the symmetries of the square, given only the adjacency matrix of the resulting graph. We give an algorithm that, if $r=n^\alpha$ for any $\alpha > 0$, with high probability reconstructs the vertex positions with a maximum error of $O(n^\beta)$ where $\beta=1/2-(4/3)\alpha$, until $\alpha \ge 3/8$ where $\beta=0$ and the error becomes $O(\sqrt{\log n})$. This improves over earlier results, which were unable to reconstruct with error less than $r$. Our method estimates Euclidean distances using a hybrid of graph distances and short-range estimates based on the number of common neighbors. We extend our results to the surface of the sphere in $\R^3$ and to hypercubes in any constant fixed dimension.\footnote{An extended abstract of this paper will be presented in ICALP-2022}
Additionally we examine the extent to which reconstruction is still possible when the original adjacency lists have had  a subset of the edges independently deleted at random.
\end{abstract}


\section{Introduction}
Graph embedding is the art of assigning a position in some smooth space to each vertex, so that the graph's structure corresponds in some way to the metric structure of that space. If vertices with edges between them are geometrically close, this embedding can help us predict new or unobserved links, devise efficient routing strategies, and cluster vertices by similarity---not to mention (if the embedding is in two dimensions) give us a picture of the graph that we can look at and perhaps interpret. In social networks, this space might correspond literally to geography, or it might be a ``latent space'' whose coordinates measure ideologies, affinities between individuals, or other continuous demographic variables (e.g.~\cite{Handcock02latentspace}). In some applications the underlying space is known; in others we wish to infer it, including the number of dimensions, whether it is flat or hyperbolic, and so on.

The literature on graph embedding is vast, and we apologize to the many authors who we will fail to cite. However, despite the broad utility of graph embedding in practice (see~\cite{zhang2021} for a recent experimental review) many popular heuristics lack rigorous guarantees. Here we pursue algorithms that reconstruct the position of every vertex with high accuracy, up to a symmetry of the underlying space.

Many versions of the reconstruction problem, including recognizing whether a graph has a realization as a geometric graph, are NP-complete~\cite{Breu98,Aspnes04,Bruck05} in the worst case. Thus we turn to distributions of random instances, and design algorithms that succeed with high probability in the instance. For many inference problems, there is a natural generative model where a ground truth structure is ``planted,'' and the instance is then chosen from a simple distribution conditioned on its planted structure. For community detection a.k.a.\ the planted partition problem, for instance, we can consider graphs produced by the stochastic block model, a generative model where each vertex has a ground-truth label, and each edge $(u,v)$ exists with a probability that depends on the labels of $u$ and $v$. Reconstructing these labels from the adjacency matrix then becomes a well-defined problem in statistical inference, which may or may not be solvable depending on the parameters of the model~(e.g.~\cite{Mossel12,moore-review,abbe-survey}). In the same spirit, a series of papers has asked to what extent we can reconstruct vertex positions from the adjacency matrix in random geometric graphs, where vertex positions are chosen independently from a simple distribution.
\subsection{Random geometric graph models.} 
Let $n$ be an integer and let $r > 0$ be real. Let $V=\{v_i\}_{i=1}^n$ be a set of points chosen uniformly at random in the square  $\left[0,\sqrt{n}\right]^2$. Then the \emph{random geometric graph} $G \in \RG$ has vertex set $V$ and edge set $E=\{ (u,v) : \|u-v\| < r \}$ where $\|u-v\|$ denotes the Euclidean distance. (We will often abuse notation by identifying a vertex with its position.)

This is simply a rescaling of the \emph{unit disk model} where $(u,v) \in E$ if the unit disks centered at $u$ and $v$ intersect, \cite{Clark90}. However, we follow previous authors in changing the density of the graph by varying $r$ rather than varying the density of points in the plane. Since the square has area $n$, the density is always $1$: that is, the expected number of points in any measurable subset is equal to its area.

It is also natural to consider a Poisson model, where the points are generated by a Poisson point process with intensity $1$. In that case the number of vertices fluctuates but is concentrated around $n$, and the local properties of the two models are asymptotically the same. We will occasionally refer to the Poisson model below. 

Note that the number of points in a region of area $A$ is binomially distributed in the uniform model, and Poisson distributed with mean $A$ in the Poisson model. In both cases, the probability that such a region of area is empty is at most $\e^{-A}$; this is exact in the Poisson model, and is an upper bound on the probability $(1-A/n)^n$ in the uniform model.

Random geometric graphs (RGGs) were first introduced by Gilbert in the early 1960s to model communications between radio stations~\cite{Gilbert}. Since then, RGGs have been widely used as models for wireless communication, in particular for wireless sensor networks. Moreover, RGG have also been extensively studied as mathematical objects, and much is known about their asymptotic properties, see for example~\cite{MathewBook, Walters}. 
One  well-known result is that $r_c=\sqrt{\log n/\pi }$ is a \emph{sharp threshold} for connectivity for $G \in \RG$ in the square in both the uniform and Poisson models: that is, for any $\eps > 0$, with high probability (i.e., with probability tending to $1$ as $n \to \infty$) $G$ is connected if $r > (1+\eps)r_c$ and disconnected if $r < (1-\eps) r_c$.

More generally, we can define RGGs on any compact Riemannian submanifold, by scattering $n$ points uniformly according to the surface area or volume. We then define the edges as $E=\{ (u,v) : \Vert u-v\Vert_g \le r \}$ where $\Vert \cdot \Vert_g$ is the geodesic distance, i.e.\ the arc length of the shortest geodesic between $u$ and $v$. On the sphere in particular this includes the cosine distance, since $\Vert u-v\Vert_g$ is a monotonic function of the angle between $u$ and $v$.
\subsection{The reconstruction problem.} 
Given the adjacency matrix $A$ of a random geometric graph defined on a smooth submanifold $M$, we want to find an embedding $\phi:V \to M$ which is as close as possible to the true positions of the vertices. We focus on the max distance $\max_v \|\phi(v)-v\|$ where we identify each vertex $v$ with its true position. (In the sequel we sometimes say ``distortion'' or ``error'' for subsets of vertices or single vertices.) 

However, if we are only given $A$, the most we can ask is for $\phi$ to be accurate up to $M$'s symmetries. In the square, for instance, applying a rotation or reflection to the true positions results in exactly the same adjacency matrix. Thus we define the \emph{distortion}  $d^*(\phi)$ as the minimum of the maximum error achieved by composing $\phi$ with some element of the symmetry group $\sym(M)$, 
\begin{equation}
\label{eq:distortion}
d^*(\phi) = \min_{\sigma \in \sym(M)} \max_{v \in V} \| (\sigma \circ \phi)(v)-v \| \, .
\end{equation}
As in previous work, our strategy is to estimate the distances between pairs of vertices, and then use geometry to find points with those pairwise distances. We focus on the case where $M = \left[0,\sqrt{n}\right]^2$ and $\Vert \cdot \Vert$ is the Euclidean distance. However, many of our results apply more generally, both in higher dimensions and on curved manifolds.

For those who enjoy group theory, if $M$ is the square then $\sym(M)$ is the dihedral group $D_8$. For the sphere, $\sym(M)$ is the continuous group $O(3)$ of all rotations and reflections, i.e., all $3 \times 3$ orthogonal matrices. 
\subsection{Our contribution and previous work.} 
An intuitive way to estimate the Euclidean distance $\|u-v\|$ in a random geometric graph is to relate it to the graph distance $d_G(u,v)$, i.e., the number of edges in a topologically shortest path from $u$ to $v$. The upper bound $\|u-v\| \le r d_G(u,v)$ is obvious. Moreover, if the graph is dense enough, then shortest paths are fairly straight geometrically and most of their edges have Euclidean length almost $r$, and this upper bound is not too far from the truth, e.g.~\cite{Bradonjic10,Ellis07,Muthu05,AriasCastro}.

As far as we know, the best upper and lower bounds relating Euclidean distances to graph distances in RGGs are given in~\cite{diaz2016relation}. In~\cite{diaz2019learning} these bounds were used to reconstruct with distortion $(1+o(1))r$ when $r$ is sufficiently large, namely if $r=n^\alpha$ for some $\alpha > 3/14$. 

However, since the graph distance $d_G$ is an integer, so the bound $\|u-v\| \le r d_G(u,v)$ cannot distinguish Euclidean distances that are between two multiples of $r$. Thus, as discussed after the statement of Theorem~\ref{thm:greedy-routing} below, the methods of~\cite{diaz2019learning} cannot avoid a distortion that grows as $\Omega(r)$. Intuitively, the opposite should hold: as $r$ grows the graph gets denser, neighborhoods get smoother, and more precise reconstructions should be possible.

We break this $\Omega(r)$ barrier by using a hybrid distance estimate. First we note that $r d_G(u,v)$ is a rather good estimate of $\|u-v\|$ if $\|u-v\|$ is just below a multiple of $r$, and we improve the bounds of~\cite{diaz2016relation} using a greedy routing analysis. Then, we combine $r d_G$ with a more precise short-range estimate based on the number of neighbors that $u$ and $v$ have in common. In essence, we use a quantitative version of the popular heuristic that two vertices are close if they have a large Jaccard coefficient (see e.g.~\cite{sarkar2011theoretical} for link prediction, and~\cite{AbrahamCKS15} for a related approach to small-world graphs). 
\begin{figure}
\centering
  \includegraphics[width=2.7in]{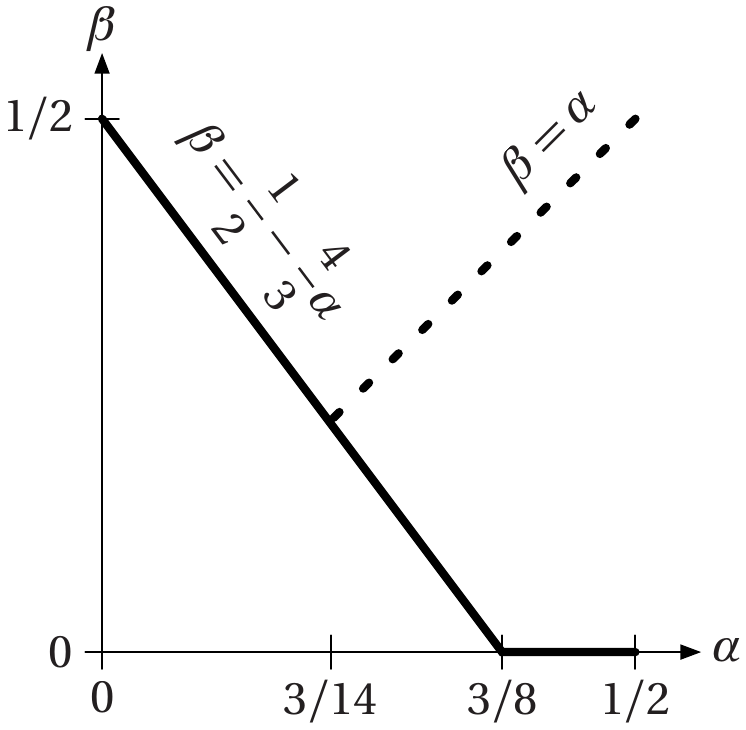}
  \caption{Our results (solid) compared to those of~\cite{diaz2019learning} (dotted). If $r=n^\alpha$, 
  our reconstruction has distortion $O(n^\beta)$ where $\beta = 1/2 - (4/3)\alpha$, except for
 $\alpha > 3/8$ where the distortion is $O(\sqrt{\log n})$. The algorithm of~\cite{diaz2019learning} applies when $\alpha > 3/14$ and gives $\beta=\alpha$, i.e. distortion $O(r)$. Our results apply for any constant $0 < \alpha < 1/2$ and give lower distortion when $\alpha > 3/14$.}
 \label{fig:comparison}
\end{figure}

As a result, we obtain distortion $d^*$ that decreases with $r$: namely, if $r=n^\alpha$ for $\alpha > 0$, then $d^*=O(n^\beta)$ where $\beta = \frac{1}{2} - \frac{4}{3}\alpha$, until for $\alpha \ge 3/8$ where $d^*=O(\sqrt{\log n})$. (Note that any $\alpha > 0$ puts us well above the connectivity threshold.) Since it uses graph distances, the running time of our algorithm is essentially the same as that of All-Pairs Shortest Paths. To our knowledge, this is the smallest distortion achieved by any known polynomial-time algorithm. We compare our results with those of~\cite{diaz2019learning} in Figure~\ref{fig:comparison}. 

We show that our results extend to higher dimensions and some curved manifolds as well. With small modifications, our algorithm works in the $m$-dimensional hypercube for any fixed $m$ (the distortion depends on $m$, but the running time does not). We also sketch a proof that it works on the surface of the sphere, using spherical geometry rather than Euclidean geometry, solving an open problem posed in~\cite{diaz2019learning}. Our techniques are designed to be easy to apply on a variety of curved manifolds and submanifolds, although we leave the fullest generalizations to future work.

We also show that our results can be extended to ``soft'' random geometric graphs, \emph{i.e.,} random geometric graphs in which a fraction of the edges are independently deleted from the adjacency list representation before it is shown to us.  Perhaps surprisingly, our techniques work very well in this setting, with only minor modifications.

We use $N(u) = \{ w : (u,w) \in E \}$ to denote the topological neighborhood of a vertex $u$, and $B(u,r)$ to denote the geometrical ball around it. Our results, as well as many of the cited results, hold \emph{with high probability} (\whp) in the random instance $G \in \RG$, i.e.\ with probability tending to $1$ as $n\to \infty$. When we consider randomized algorithms, the probability is over both $\RG$ and the randomness of the algorithm.
\subsection{Other related work.} 
In the statistics community there are a number of consistency results for maximum-likelihood methods (e.g.~\cite{shalizi2019consistency}) but it is not clear how the accuracy of these methods scales with the size or density of the graph, or how to find the maximum-likelihood estimator efficiently. There are also results on the convergence of spectral methods, using relationships between the graph Laplacian and the Laplace-Beltrami operator on the underlying manifold (e.g.~\cite{araya2019}). This approach yields bounded distortion for random dot-product graphs in certain regimes.

We assume that parameters of the model are known, including the underlying space and its metric structure (in particular, its curvature and the number of dimensions). Thus we avoid questions of model selection or hypothesis testing, for which some lovely techniques have been proposed (e.g.~\cite{Parthasa17,lubold2021identifying,Bubeck16}). We also assume that the parameter $r$ is known, since this is easy to estimate from the typical degree.

\subsection{Organization of the paper.}

In Section~\ref{sec:prelim} we present some concentration results and define the concept of a \emph{deep} vertex.  Intuitively, a vertex is deep if it is more than $r$ from the boundary of the square, so that its ball of potential neighbors is entirely in the interior. However, since we are only given the adjacency matrix of the graph, we base our definition on the number of vertices two steps away from $v$, and show that these topological and geometric properties are closely related.

Section~\ref{sec:dist} shows that we can closely approximate Euclidean distances $\|u-v\|$ given the adjacency matrix whenever $v$ is deep. We do this in two steps: we give a precise short-range estimate of $\|u-v\|$ when $u,v$ have graph distance $2$, and also a long-range estimate that uses the existence of a greedy path. By ``hybridizing'' these two distance estimates, switching from long to short range at a carefully chosen intermediate point, we obtain a significantly better estimate of $\|u-v\|$ than was given in~\cite{diaz2016relation}. We believe these distance estimation techniques may be of interest in themselves. 

In Section~\ref{sec:reconstruction}, we use this new estimate of Euclidean distances to reconstruct the vertex positions up to a symmetry of the square, by starting with a few deep ``landmarks'' and then triangulating to the other vertices.  This gives smaller distortion than the algorithm in~\cite{diaz2019learning}, achieving the scaling shown in Figure~\ref{fig:comparison}.

In Section~\ref{sec:extensions}, we extend our method to random geometric graphs in the $m$-dimensional hypercube and the $m$-dimensional sphere for fixed $m$.

In Section~\ref{sec:missing}, we discuss extensions to ``soft'' random geometric graphs, in which a fraction of the edges have been randomly deleted from the graph before its adjacency matrix is shown to us.

Finally, in section~\ref{sec:conclusion} we give some conclusions and pose a set of open research problems.

\section{Preliminaries}\label{sec:prelim}
\subsection{Concentration of Measure}

We will use the following version of Chernoff's bound, which applies to both Poisson and binomial random variables, and in which the statement of the result has been ``inverted'' to describe a confidence interval for the outcome of the experiment, for a given confidence parameter $\delta.$

\begin{lemma} \label{lem:chernoff-confidence-interval}
Let $X$ be either a Poisson or binomial random variable with mean $\mu$. Then, for every $\delta > 0$, 
\[
\Prob{|X - \mu| > \max \left\{ 3 \log \frac{1}{\delta} ,  \sqrt{3 \mu \log \frac{1}{\delta} } \right\} } \le \delta.
\]
\end{lemma}

We will apply this lemma to derive a confidence interval for the number of vertices within a given region of the square $\SRn = [0, \sqrt{n}] \times [0, \sqrt{n}]$.
 
\begin{corollary} \label{cor:chernoff-for-area}
Let $R \subseteq \SRn$ be a region of area $A$.  Then, for any $C \ge 1$, with probability
at least $1 - 1/n^C$, the number $X$ of vertices in $R$ satisfies
\[
|X - A| \le \max \left\{ 4 C \log n, 2 \sqrt{C A \log n} \right\}.
\]
Furthermore, for sufficiently large $n$, this bound still holds even if we have conditioned on the positions of a constant number of vertices, and $R$ is allowed to depend on these.
\end{corollary} 
\subsection{Deep Vertices}
In what follows we assume that we are working with random geometric graphs in the square $\SRn$.
Because some of our arguments will break down for vertices near the boundary of this square, it will be useful to have an easy way to tell these vertices apart from the rest.  To this end, we introduce the following notion of ``deep" vertex.

\begin{definition}\label{def:deep-vert}
Let $r$ be fixed.  We say that a vertex $v \in V$ is \emph{deep} if at least $11 r^2$ vertices have graphical distance $2$ or less from $v$.
\end{definition}

Note that being deep is a property of the graph, rather than its embedding in the plane.  This will be important for our algorithms, which are only given access to the adjacency matrix.  On the other hand, the following observation shows that being deep almost surely implies being fairly far from the boundary,

\begin{lemma} \label{lem:deep-are-safe}
There exist constants $C_1, C_2 > 0$ such that, for every $n \ge 1, r \ge C_1 \sqrt{\log n}$,
with probability at least $1 - C_2/n^2$,  every deep vertex of $G$ is located in the square $[r, \sqrt{n} - r]^2$, that is, the entire ball of radius $r$ centered at the vertex is contained in $\SRn$.
\end{lemma}

\begin{proof}
Suppose $v$ is a vertex located within distance $r$ of the boundary of $\SRn$.  Let $B=B(v,2r)$ be the ball of radius $2r$ centered at $v$.  Then $B \setminus \SRn$ contains a circular segment of height at least $r$, and therefore central angle at least $2\pi/3$, and area at least $\left( \frac{4}{3} \pi - \sqrt{3} \right) r^2$.
It follows that $B \cap \SRn$ has area at most $\left( \frac{8}{3} \pi + \sqrt{3} \right) r^2 \approx 10.1 r^2$.

Since every vertex within distance 2 of $v$ is located within $B \cap \SRn$, in expectation the number of such vertices is at most $10.1 r^2$.  Since each vertex other than $v$ is independently in $B$ with probability
$\frac{\vol(B \cap \SRn)}{\vol(\SRn)}$, we can apply a Chernoff bound to conclude that, with high probability, fewer than $11 r^2$ vertices are in $B \cap \SRn$, and therefore $v$ is not deep.
\end{proof}

On the other hand, we will also need a large supply of deep vertices.  Fortunately, this follows from another observation, which we now state and prove.

\begin{figure} 
    \centering
    \includegraphics[width=2in]{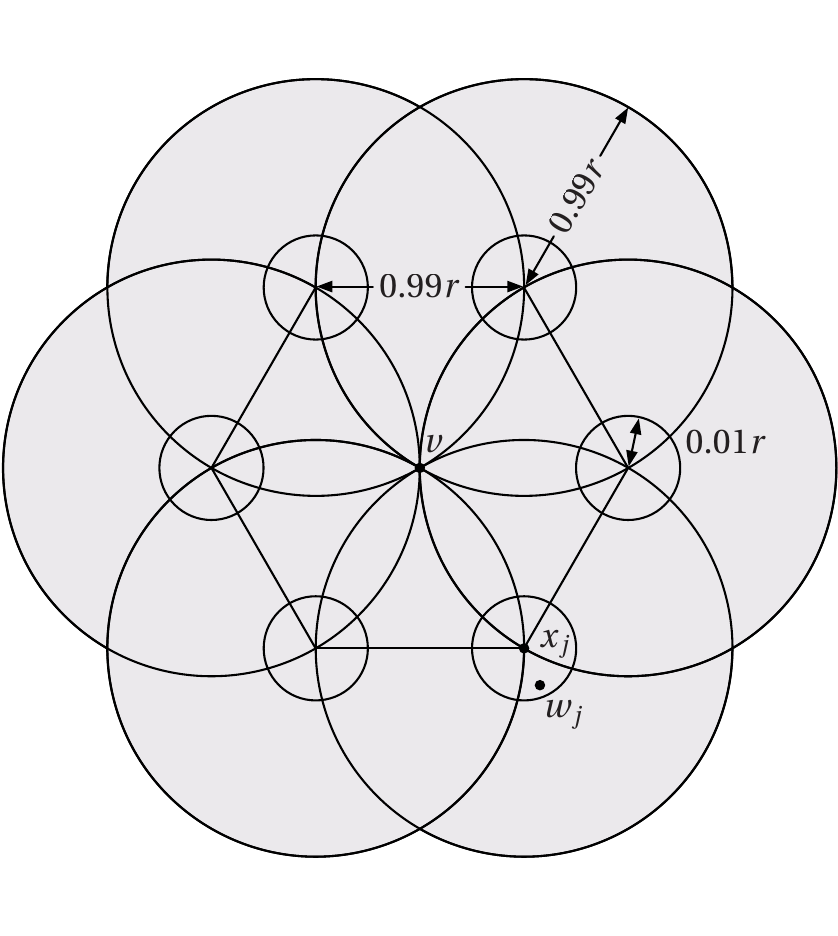}
    \caption{The ``flower'' of the proof of Lemma~\ref{lem:lots-of-deep}. The hexagon is centered at $v$, its side length is $0.99r$, and its corners are $x_1,\ldots,x_6$. With high probability, each of the balls of radius $0.01r$ (not shown to scale) centered on the $x_j$ contains at least one point $w_j$; these are neighbors of $v$ since $0.99r+0.01r=r$. Moreover, each $w_j$ is a neighbor of every point in the ball of radius $0.99r$ centered at $x_j$. Therefore, every point in the gray region has graph distance at most $2$ from $v$.} 
\label{fig:flower}
\end{figure}

\begin{lemma} \label{lem:lots-of-deep}
There exist constants $C_1, C_2 > 0$ such that, if $r \ge C_1 \sqrt{\log n}$, with probability at least $1 - C_2/n^2$ every vertex located in $[2r, \sqrt{n} - 2r]^2$ is deep. 
\end{lemma}

\begin{proof}
Let $v$ be located in $[2r, \sqrt{n}-2r]^2$.  Let $x_1, \dots, x_6$ be the corners of a regular hexagon whose center is at $v$, and whose side length is $0.99 r$.  

We claim that, with high probability, there exist vertices $w_1, \dots, w_6 \in V$ such that $\|w_j - x_j\| \le 0.01 r$ for each $j$. To see this, note that the ball of radius $0.01 r$ centered at $x_j$ has area $A=0.0001 \pi r^2$. Since $r \ge  C_1 \sqrt{\log n}$, the probability that a given one of these balls is empty is at most $\e^{-A} \le 1/6n^3$ if $C_1$ is sufficiently large. We then apply a union bound over these six events. 

By the triangle inequality, $\|w_j - v\| \le 0.99r+0.01r = r$, so these $w_j$ are neighbors of $v$. To see that the $w_j$ in turn have many neighbors, which have graph distance $2$ from $v$, we apply Corollary~\ref{cor:chernoff-for-area} to the ``flower'' formed by the union of the six balls of radius $0.99 r$ centered at each of the $x_j$'s, shown in gray in Figure~\ref{fig:flower}. Any point in one of these balls is a neighbor of the corresponding $w_j$, again by the triangle inequality. Thus this region contains only vertices at distance $2$ or less from $v$. 

An easy exercise in geometry shows that the area of this flower equals
$(0.99)^2 (2 \pi + 3 \sqrt{3}) r^2 \approx 11.25 r^2$. Corollary~\ref{cor:chernoff-for-area} then implies that, with high probability, $\left| \bigcup_{j=1}^6 N(w_j) \right| \ge 11 r^2$ and hence $v$ is deep. Finally, we take the union bound over all vertices in $[2r,\sqrt{n}-2r]^2$, of which there are at most $n$.
\end{proof}

\section{Distance Estimation: Breaking the $\Omega(r)$ barrier} 
\label{sec:dist}
\subsection{Estimating short-range distances: Lunes and lenses}
In this section we show how to estimate the Euclidean distance between two vertices $v,u$ such that $d_G(v,u) = x \le 2r$. In Figure~\ref{fig:distance2}, we have two  balls  centered at $v$ and $u$ and with radius  $r$, $B(v,r)$ and $B(u,r)$. 
In the case $x<r$  we will be dealing 
with the {\em lune} defined by each half of the symmetric difference between $B(u,r)$ and $B(v,r)$ (that is the case reflected in the figure). In the case that $r\le x<2r$ , we will be dealing with the {\em lens} defined by the intersection of the balls.
\begin{figure}
    \centering
    \includegraphics[width=2.1in]{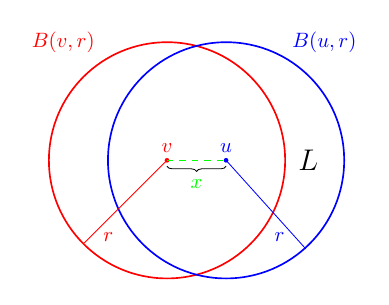}
    \caption{We can estimate the Euclidean distance $\Vert u-v \Vert = x$ of two vertices with $d_G(u,v) \le 2$ using the area of the lune $L=B(u,r)\setminus B(v,r) \neq \emptyset$.
     Denoting this area $F(x)$, we can estimate $x$ by applying the inverse $F^{-1}$ to the number of points in $N(v) \setminus N(u)$. For the case $x>r$ we should talk about a lense
     $B(u,r)\cap B(v,r)$.}
    \label{fig:distance2}
\end{figure}

Let $0 \le x \le 2r$. We define the function $F(x) = F(x,r)$ to be one half of the area of the symmetric difference of two balls of radius $r$ whose centers are at distance $x$, that is, the area of the lune $B((x,0),r) \setminus B((0,0),r)$, here for concreteness $u$ and $v$ are located at $(0,0)$ and $(x,0)$) respectively.
Equivalently, $F$ is given by the following formula:
\begin{equation}
\label{eq:F}
F(x) = \pi r^2 - 2 r^2 \arccos \frac{x}{2r} 
+ \frac{x}{2} \sqrt{4r^2 - x^2}
\end{equation}
or the following differential equation:
\begin{equation}
\label{eq:F-diffeq}
F'(x) = \sqrt{4r^2 - x^2},
\end{equation}
with boundary condition $F(0) = 0$.  This also implies that, on the interval $(0,r)$, we have
$F'(x) \ge r \sqrt{3}$.

We will need the following easy fact about $F$.
\begin{lemma} \label{lem:concave-MVT}
Let $0 \le x_1 \le x_2 \le 2r$. Then 
\[
F(x_2) \ge F(x_1) + (x_2-x_1) \frac{F'(x_1) + F'(x_2)}{2}.
\]
\end{lemma}
\begin{proof}
Since $F'(x) = \sqrt{4r^2 - x^2}$, differentiating twice more, we find that 
\[
F'''(x) = \frac{-4r^2}{(4r^2-x^2)^{3/2}} \le 0 \, ,
\]
so $F'$ is concave.  

Sample $X$ uniformly from the interval $[x_1,x_2]$, then sample $Y$ uniformly from the two endpoints, $\{x_1,x_2\}$, 
such that its conditional expectation is $X$.  Note that this implies that $Y$ is uniform over $\{x_1,x_2\}$.
Then we have:
\begin{align*}
\frac{F(x_2) - F(x_1)}{x_2-x_1} &= \Exp(F'(X)) & \mbox{by the Fundamental Theorem of Calculus} \\
&= \Exp(F'(\ExpCond{Y}{X})) & \mbox{by definition of Y} \\
&\ge \Exp(\ExpCond{F'(Y)}{X}) &\mbox{by Jensen's Inequality} \\
&= \Exp(F'(Y)) &\mbox{by the Law of Total Probability} \\
&= \frac{F'(x_1) + F'(x_2)}{2} & \mbox{since $Y$ is uniform over $\{x_1,x_2\}$.}
\end{align*}
Rearranging, this completes the proof.
\end{proof}

Now, for a random geometric graph, $G$, and two vertices, $v, w$, the neighbors of $v$ that are non-neighbors
of $w$ are exactly the vertices in the lune $B(v,r) \setminus B(w,r)$, which has area $F(\|v-w\|)$.  Assuming $v$ is deep, this lune is
entirely within the domain $\SRn$, so, intuitively, $|N(v) \setminus N(w)|$ should be close to $F(\|v - w\|)$.
With this in mind, we define
\begin{equation}
\label{eq:dtilde}
\dtilde(v,w) =  F^{-1}\left( \min\{ |N(v) \setminus N(w)|, F(r) \} \right) ,
\end{equation}
with the hope that $\dtilde(v,w)$ will be a good approximation to the true Euclidean distance, $\|v-w\|$, most of the time. 
Note that $\dtilde(v,w)$ is a function of the adjacency matrix and $r$.

The following Lemma makes the above intuition precise.
\begin{lemma} \label{lem:adjacent-distances}
Let $G$ be a random geometric graph with parameters $r,n$. With probability at least $1 - 1/n^2$ we have, for all $(v,w) \in E(G)$ such that $v$ is deep, 
\begin{equation}
\label{eq:dtilde-good}
\left| \| v - w \| - \dtilde(v,w) \right| \le  100 \max \left\{ \frac{\log n}{r}, \sqrt{\frac{\|v-w\| \log n}{r}} \right\}.
\end{equation}
\end{lemma}
\begin{proof}
Fix a pair of vertices $v,w$, and condition on their positions.  
Let $d = \| v - w \|$.
Assume that $v \in [r, \sqrt{n}-r]^2$, so that
the ball of radius $r$ centered at $v$ is entirely within the domain, $\SRn$.
By Lemma~\ref{lem:deep-are-safe}, with high probability, this assumption will not
cause us to miss any deep vertices.  Also assume that $d \le r$, so that $v$ and $w$ are neighbors,
since otherwise there is nothing to prove.

Observe that $N(v) \setminus N(w)$ consists of $w$, together with exactly those vertices
located within the lune $L = B(v,r) \setminus B(w,r)$. 
Since $L$ has area $F(d)$ and is contained within the domain, $\SRn$, 
it follows by Corollary~\ref{cor:chernoff-for-area} that
with probability at least $1 - \frac{1}{n^4}$, the number of vertices, $X$, in $L$
satisfies
\[
|X - F(d)| \le \max \left\{ 16 \log n, 4 \sqrt{F(d) \log n} \right\},
\]
and hence
\[
\left| |N(v) \setminus N(w)| - F(d) \right| \le 1 + \max \left\{ 16 \log n, 4 \sqrt{F(d) \log n} \right\}.
\]
Now, since the derivative of $F$ satisfies
\[
F'(x) \ge r \sqrt{3}
\]
for $0 \le x \le r$, it follows by the Inverse Function Theorem that $0 \le (F^{-1})'(y) \le \frac{1}{r \sqrt{3}}$
for $0 \le y \le F(r)$, and hence that 
\[
|\dtilde(v,w) - d| \le \frac{1 + 16 \left( \log n + \sqrt{F(d) \log n} \right)}{r \sqrt{3}}
\, .
\]

Since $F(d) \le 2 d r$, this implies

\[
|\dtilde(v,w) - d| \le \frac{1 + \max \left\{ 16 \log n, 4 \sqrt{2dr \log n} \right\}}{r \sqrt{3}} 
\, ,
\]
which is \eqref{eq:dtilde-good}, only with slightly better constants. Taking a union bound over all possible choices of $v,w$ completes the proof.
\end{proof}

\begin{remark}
Notice equation \ref{eq:dtilde-good}  is only useful when $r$ is at least a constant factor greater than the connectivity threshold. 
In particular, for $r \le 100 \sqrt{\log n}$, the conclusion of Lemma~\ref{lem:adjacent-distances} becomes trivial,
because the right-hand side of \eqref{eq:dtilde-good} exceeds $r$.  
\end{remark}

So far we have only defined $\dtilde(v,w)$ for adjacent pairs of vertices where one is deep. We now extend this definition to pairs of vertices at graphical distance 2. Once again, we will require that at least one of them is deep. 

Let $v, w \in V$ be such that $v$ is deep, and $d_G(v,w)=2$. Let $d = \|v-w\|$. Note that $r < d \le 2r$. This time, rather than using the number of vertices in $B(v,r) \setminus B(w, r)$ to estimate the distance, we will instead use the intersection of the balls, which has area $\pi r^2 - F(d)$. 

$N(v) \cap N(w)$ is the set of vertices of $V \setminus \{v,w\}$ that are located within the lens-shaped intersection $B(v,r) \cap B(w,r)$. Since $v$ is deep,  $B(v,r) \cap B(w,r) \subset \SRn$.  Therefore, conditioned on the positions of $v$ and $w$,
 \[
\Exp\left[ \left|N(v) \cap N(w)\right| \right] = \frac{(n-2) \left(\pi r^2 - F(d)\right)}{n} \,.
\]
With this in mind, we will write down the distance estimate as
\begin{equation}
\label{eq:dtilde2}
\dtilde(v,w) =  F^{-1}\left( \max\left\{ F(r), \pi r^2 - |N(v) \cap N(w)| \right\}\right)
\end{equation}
and we will now prove that in this case also,  $\dtilde(v,w) \approx \|v-w\|$ with high probability.

\begin{lemma} \label{lem:2-apart-distances}
Let $G = (V,E)$ be a random geometric graph with parameters $r,n$.  Then, with probability at least $1 - 1/n^2$ we have, for all $v,w \in V$ such that $v$ is deep and $d_G(v,w) = 2$, 
\begin{equation} \label{eq:dtilde-good-near-2r}
\left| \| v - w \| - \dtilde(v,w) \right| \le  100 \max \left\{ \frac{(\log n)^{2/3}}{r^{1/3}}, 
\left( \frac{2r - \|v-w\|}{r} \right)^{1/4} \sqrt{\log n} \right\}.
\end{equation}
\end{lemma}

\begin{remark}
Similarly to Lemma~\ref{lem:adjacent-distances}, we note that 
\eqref{eq:dtilde-good-near-2r} is
only useful when $r$ is at least a constant factor greater than the connectivity threshold. 
In particular, for $r \le 100 \sqrt{\log n}$, the conclusion of Lemma~\ref{lem:adjacent-distances} becomes trivial,
because the right-hand side of \eqref{eq:dtilde-good-near-2r} exceeds $r$.  
\end{remark}

\begin{proof}
Fix a pair of vertices $v,w$, and condition on their positions. Let $d = \| v - w \|$.  
Assume that $B(v,r) \subset \SRn$, since otherwise $v$ is almost surely not deep, so we would have nothing to prove. Also assume that $r < \|v-w\| \le 2r$, since otherwise $d_G(v,w)$ cannot be 2, and we would again have nothing to prove.

Observe that, since $d(v,w) > r$, $N(v) \cap N(w)$ consists of exactly those vertices located within the lens $L = B(v,r) \cap B(w,r)$. Since $L$ has area $\pi r^2 - F(d)$ and is contained within the domain, $\SRn$, it follows by Corollary~\ref{cor:chernoff-for-area} that
with probability at least $1 - \frac{1}{n^4}$, the number of vertices, $X$, in $L$ satisfies
\[
|X - (\pi r^2 - F(d))| \le \max \left\{ 16 \log n, 4 \sqrt{ (\pi r^2 - F(d)) \log n} \right\},
\]
and hence
\[
\left| |N(v) \cap N(w)| - (\pi r^2 - F(d)) \right| \le \max \left\{ 16 \log n, 4 \sqrt{(\pi r^2 - F(d)) \log n} \right\}.
\]

By Lemma~\ref{lem:concave-MVT} applied with $\{x_1,x_2\} = \{d, \dtilde\}$, we have
\[
|F(d) - F(\dtilde)| \ge |d - \dtilde| \frac{F'(d) + F'(\dtilde)}{2}
\]
\newcommand{\epstilde}{\widetilde{\eps}}
We change variables, using $\varepsilon = 2r - d$, $\widetilde{\varepsilon} = 2r - \dtilde$, and apply the definition of $F'$
to find
\begin{align*}
|d - \dtilde| &\le \frac{2 |F(d) - F(\dtilde)|}{F'(d) + F'(\dtilde)} \\
&= \frac{2 | F(d) - \max\{F(r), \pi r^2 - |N(v) \cap N(w)| \} | }{F'(d) + F'(\dtilde)} \\
&= \frac{2 | F(d) - \max\{F(r), \pi r^2 - |N(v) \cap N(w)| \} | }{ \sqrt{4r \eps - \eps^2} + \sqrt{4r \epstilde - \epstilde^2} } \\
&\le \frac{2 \max \left\{ 16 \log n, 4 \sqrt{(\pi r^2 - F(d)) \log n} \right\} }{ \sqrt{4r \eps - \eps^2} + \sqrt{4r \epstilde - \epstilde^2} } \\
&\le \frac{2 \max \left\{ 16 \log n, C \sqrt{ \eps^{3/2} r^{1/2} \log n} \right\} }{ \sqrt{4r \eps - \eps^2} + \sqrt{4r \epstilde - \epstilde^2} } 
\\
&= O\!\left( \frac{\eps}{r} \right)^{1/4} \sqrt{\log n} \, ,
\end{align*}
assuming $\eps^3 r > C (\log n)^2$.  For the other case, we may as well treat $\eps$ as zero, in which case we get
\begin{align*}
|d - \dtilde| &= \epstilde \le \frac{16 \log n}{\sqrt{4r \epstilde - \epstilde^2}} \\
&\approx \frac{16 \log n}{\sqrt{3r \epstilde}}, \\
\end{align*}
which implies
\[
\epstilde \le \left(\frac{16 \log n}{\sqrt{3r}}\right)^{2/3},
\]
as desired.
\end{proof}

The results of Lemmas~\ref{lem:adjacent-distances} and \ref{lem:2-apart-distances} are summarized in the following theorem.

\begin{theorem} \label{thm:short-dis}
Let $G = (V,E)$ be a random geometric graph with parameters $r,n$, where $r \ge 100 \sqrt{\log n}$. 
With probability at least $1 - 2/n^2$ we have, for all vertices $v \ne w$ such that 
$d_G(v,w) \le 2$ and $v$ is deep, 
\begin{equation}
\label{eq:dtilde-good-combined-summary}
\left| \Vert  v - w \Vert  - \dtilde(v,w) \right| \le  100 \eta (\Vert v-w\Vert ) \sqrt{\log n} \, ,
\end{equation}
where $\eta: [0,2r] \to [0,1]$ is defined by
\[
\eta(x) = \begin{cases}
\frac{\sqrt{\log n}}{r} &\mbox{for $0 \le x \le \frac{\log n}{r}$}, \\
\sqrt{\frac{x}{r}}&\mbox{for $\frac{\log n}{r} \le x \le r$}, \\
\left(\frac{2r-x}{r}\right)^{1/4} & \mbox{for $r \le x \le 2r - \frac{(\log n)^{2/3}}{r^{1/3}}$}, \\
\frac{(\log n)^{1/6}}{r^{1/3}} & \mbox{for $2r - \frac{(\log n)^{2/3}}{r^{1/3}} \le x \le 2r$}.
\end{cases}
\]
\end{theorem}

\subsection{Estimating long-range distances}

Next we show a fairly tight relationship between geometric and topological distance for all pairs of vertices, including distant ones. This is a slightly sharper version of~\cite[Thm 1.1]{diaz2016relation}.  The main difference in the proof is that, where before, a short path between two given 
vertices is found by finding vertices close to a straight line between the endpoints, our proof instead analyzes a greedy algorithm generating a path that may deviate further from the straight line.
We start with the following geometrical lemma. 

\begin{lemma} \label{lem:lens-width-area}
Let $B_1$, $B_2$ be two overlapping Euclidean balls in $\R^2$ of radius $r_1$ and $r_2$ respectively, and let $d$ be the distance between their centers. Consider the lens $L = B_1 \cap B_2$.  Let $\delta$ denote the width of $L$, that is, 
\[
\delta = \min\{ r_1 + r_2 - d, 2r_1, 2r_2\} \, .
\]
Then the area $A$ of $L$ satisfies
\[
A = \Theta\!\left( \delta^{3/2} \min\{r_1,r_2\}^{1/2} \right) \, .
\]
\end{lemma}

\begin{proof}
There are many ways to see this, and we prove an $m$-dimensional generalization in Lemma~\ref{lem:lens-width-volume} below. Here we give a proof using the function $F(x)$ from the previous section. 

Let $r=\min\{r_1,r_2\}$. Without loss of generality, we can assume $\delta < r/10$. Otherwise, $L$ includes a constant fraction of the ball of radius $r$, and $A=\Theta(\delta^{3/2} \,r^{1/2}) = \Theta(r^2)$ as stated.

Now note that $A$ is between $1/2$ and $1$ times the area of a lens of width $\delta$ formed by two balls $B_1,B_2$ of radius $r$ whose centers are $x=2r-\delta$ apart. This lens is $B_1$ minus the lune $B_1 \setminus B_2$, so its area is
\[
A_\between(\delta) = \pi r^2 - F(2r-\delta) \, .
\]
where $F$ is defined as in~\eqref{eq:F}. From~\eqref{eq:F-diffeq} we have the following differential equation for $A_\between$, 
\[
A'_\between(\delta) 
= -F'(2r-\delta) 
= \sqrt{4r^2-x^2} 
= \sqrt{4 \delta r - \delta^2} 
= \Theta(\delta^{1/2} \,r^{1/2}) \, .
\]
Integrating this from the boundary condition $A_\between(0)=0$ to a given $\delta$ gives $A=\Theta(\delta^{3/2} \,r^{1/2})$ and completes the proof.
\end{proof}
\begin{remark}
Note in the definition of $\delta$ the cases $2r_1$ and $2r_2$ refer to the situation when one ball is totally inside the other ball. In our   random geometric graphs, all the neighborhood 
balls have the same radius. But as it will be the case in the proof of Theorem~\ref{thm:greedy-routing}, we also argue about lunes between  balls with different  radii,  therefore the need to define  $\delta$ in full generality.
\end{remark}

The main result of this section is the following theorem,
\begin{theorem}
\label{thm:greedy-routing}
There exist absolute constants $C_1, C_2, C_3$ such that, for all $n \ge 1$, all $r \ge C_1 \sqrt{\log n}$, with probability at least $1 - C_2/n^2$, all pairs of vertices $u,v$ satisfy
\begin{equation}
\left\lceil \frac{\| u-v \|}{r} \right\rceil
\le d_G(u,v) 
\le \left\lceil \frac{\|u-v\| + \kappa}{r}  \right\rceil \, ,
\label{eq:lenses}
\end{equation}
where 
\begin{equation}
\label{eq:kappa}
\kappa = C_3 \left( \frac{\|u-v\|}{r^{4/3}} + \frac{\log n}{r^{1/3}} \right) \, . 
\end{equation}
\end{theorem}

\begin{figure}
  \centering
  \includegraphics[width=4.5in]{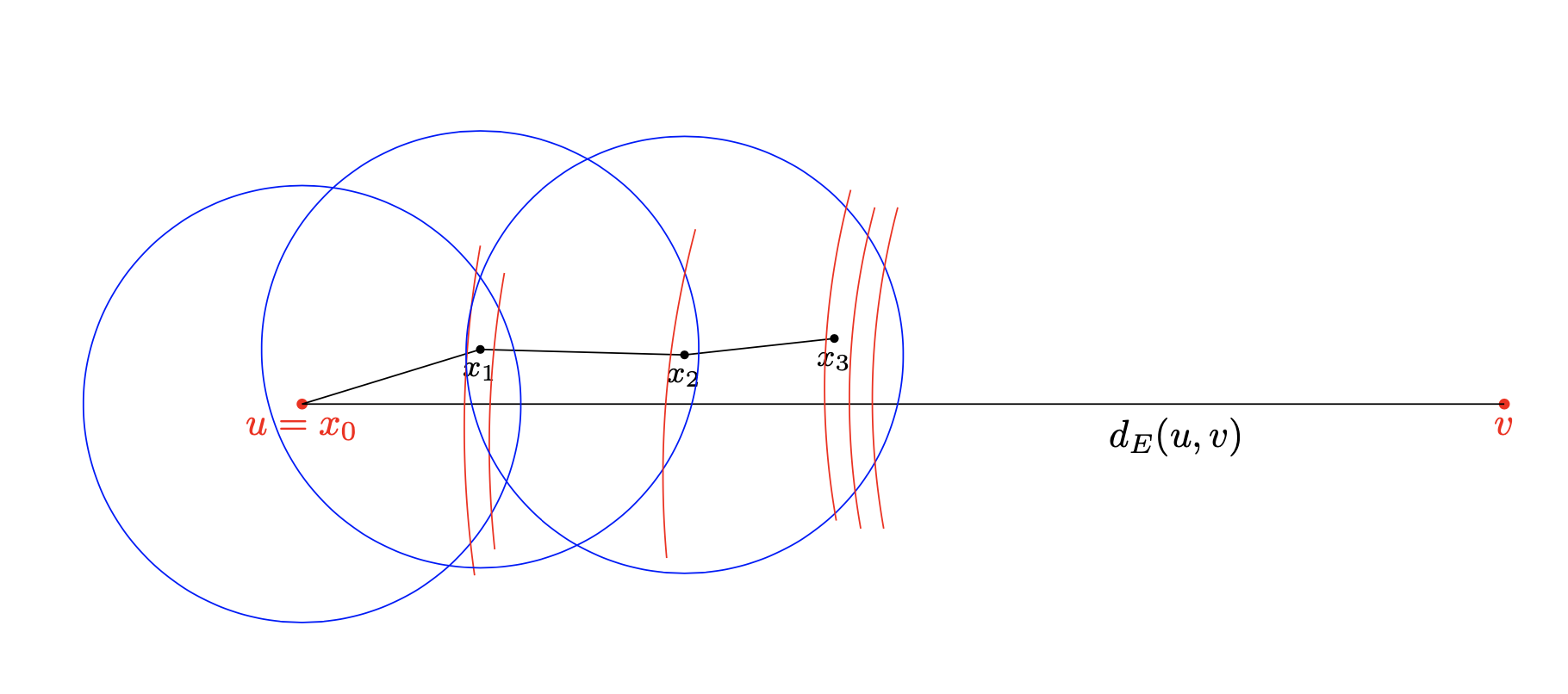}
  \caption{The greedy routing analysis of Theorem~\ref{thm:greedy-routing}. At each step we go from $x_i$ to the neighbor $x_{i+1}$ closest to $v$. In the analysis,
 we consider the intersections of $x_i$'s neighborhood with balls centered at $v$, with the radii of the latter chosen so that these intersections have area $\ln 2$, $2 \ln 2$, $3 \ln 2$, and so on. Each of these intersections contains a point with constant probability, so that most steps make significant progress towards $v$.}   
  \label{fig:Longd}
\end{figure}

\begin{proof}
The first inequality is easy: since every edge has Euclidean length at most $r$, we have $\|u - v\| \le d_G(u,v) r$ by the Triangle Inequality. Since $d_G(u,v)$ is an integer, rounding it up does not change its value, so
\[
\left\lceil \frac{\| u-v \|}{r} \right\rceil \le d_G(u,v) \, .
\]

We now turn our attention to the second  inequality of \eqref{eq:lenses}, the upper bound on $d_G(u,v)$.  
In an attempt to find a short path from $u$ to $v$, we consider the following greedy algorithm, 
(see Fig~\ref{fig:Longd}). 
Let $x_0 = u$.  For $i \ge 1$, we define $x_{i+1}$ to be the neighbor of $x_i$ that has minimal Euclidean distance to $v$ (note that $x_{i+1}$ is unique with probability $1$).  The algorithm terminates if no neighbor of $x_i$ is closer to $v$ than $x_i$ is. Hopefully this is because $x_i = v$, in which case we say there exists a greedy path from $u$ to $v$.  However, it may instead happen that the greedy algorithm gets stuck in a local minimum, and never reaches $v$.

We will prove that, with probability $1 - O(n^{-3})$, a greedy path exists from $u$ to $v$, and its length is at most the desired upper bound.  Taking a union bound over all pairs $u,v$ completes the proof. 
Notice this greedy algorithm has nothing to do with our reconstruction algorithm: it is purely for our analysis, i.e., to prove that short paths probably exist.

We will discuss our solution in terms of a sequence of independent fair coin flips.
These coin flips are the outcomes of a sequence of experiments of the form 
``Is $\|x_{i+1} - v \| \le a$?'', where $x_1, \dots, x_{\ell}$ is our greedy walk from $u$ to $v$.
The values $a$ are chosen adaptively to make each of these coin flips fair, conditioned on previously revealed steps of the greedy walk, and on previous coin flips.

Specifically, assume that $x_i$ has just been revealed.  Our first coin flip asks whether $x_{i+1}$ lies in the lens $L_1$  defined as the intersection of the ball of radius $r$ centered at $x_i$ with the ball of radius $a_1$ centered at $v$, with $a_1$ chosen so that $L_1$ has area $\ln 2$. Working in the Poisson model, the probability that $L_1$ contains at least one vertex is exactly $1/2$.  If the coin comes up ``no'', i.e., $L_1$ is empty, we ask the same question, while increasing the radius of the ball around $v$ to $a_2$, and then $a_3$, and so on, with $a_t$ chosen so that the resulting lens $L_t$ has area $t \ln 2$. Each time we increment $t$, we gain an additional region $L_t \setminus L_{t-1}$ of area $\ln 2$, which corresponds to another coin flip. We continue until the $t$th coin comes up ``yes'', i.e., $L_t$ contains at least one point. At that point we know that  $x_{i+1} \in L_t \setminus L_{t-1}$ and therefore $\|x_{i+1}-v\| \le a_t$ as shown in Figure~\ref{fig:coins}. We then reveal  $x_{i+1}$ (note that it might be one of several points in that region). 

\begin{figure}
    \centering
    \includegraphics[width=3in]{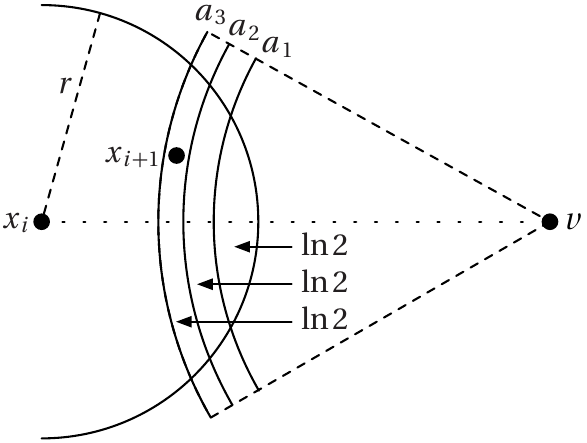}
    \caption{The analysis of greedy routing in  Theorem~\ref{thm:greedy-routing}. We consider the intersections of the neighborhood of the current point $x_i$ with balls centered at $v$, with the radius of the latter chosen so that these intersections have area $\ln 2$, $2 \ln 2$, $3 \ln 2$, and so on. Each additional region of area $\ln 2$ is nonempty with probability $1/2$, creating a fair coin flip. When this coin gives ``yes'', i.e., that region is nonempty, we reveal $x_{i+1}$ in that region.}
    \label{fig:coins}
\end{figure}

We then repeat the whole process with $x_{i+1}$ in place of $x_i$.  This continues until $x_i$ is within the disk of radius $r$ around $v$, 
after which $x_{i+1}=v$ and we're done.

Note that the above argument breaks down if the sequence of lenses examined in stage $i+1$ overlaps the lenses examined in stage $i$.  In this case, since the planar regions being examined aren't disjoint, the second sequence of coin flips is no longer conditionally independent of the first.  Fortunately, this can only happen if there is a very long sequence of consecutive ``no'' answers in stage $i$, which is unlikely in a sequence of fair coin flips.  Specifically, we are only worried about the event that $\|x_i - v\| \ge \|x_{i-1} - v\| - r/2$, which, since the area of the lens of width $r/2$ with centers at $x_{i-1}$ and $v$ and radii $r$ and $\|x_i-v\|-r/2$, respectively, is $\Theta(r^2)$ by Lemma~\ref{lem:lens-width-area}.  Choosing the constants carefully, we can ensure that the probability of this lens being empty is $O(1/n^4)$, and hence with high probability it doesn't happen at any stage of the 
greedy algorithm.

The next step in the analysis is to relate the number of ``no'' answers received in the aforementioned sequence of fair coin flips to the progress made by the greedy algorithm.  Note that each sequence of $k-1$ ``no'' answers followed by one ``yes,'' corresponds to $x_{i+1}$ being in the lens of area $k \ln 2$, which has, by Lemma~\ref{lem:lens-width-area}, width $\Theta(k^{2/3} r^{-1/3})$, which is $O(k r^{-1/3})$. Thus 
\[
\|x_{i+1}-v\| 
\le a_k = \|x_i-v\| - r  + O(k r^{-1/3}) \, .
\]
It follows that, if the total number of coins flipped on the journey from $u$ to $v$ is $m$, we necessarily have $\|u-v\| \ge r (d_G(u,v) - 1) - O(m r^{-1/3})$.   

Setting $m = \frac{4 \| u - v \|}{r} + 48 \log n$,  Chernoff's bound tells us that 
the probability of getting fewer than $m/4$ heads in a sequence of $m$ coin flips is at most
$\exp(-m/16) = O(n^{-3})$.  The second kind of bad event is getting a sequence of more than 
$4 \log n$ consecutive tails.  Since the chance of having a particular set of $k$ coin flips be all tails
is $2^{-k}$, and there are at most $m = O(\sqrt{n}/r)$ possible positions for a run to start, a
union bound shows that the combined probability of having such a run is $O(n^{-4} \sqrt{n}/r) = O(n^{-3})$.

Putting all of the above together, except for an $O(n^{-3})$ probability of failure, we have
\begin{align*}
\|u-v\| &\ge  r (d_G(u,v) - 1) - O(m r^{-1/3}) \\
&= r(d_G(u,v) - 1) - O\left(  \frac{\|u-v\|}{r^{4/3}} + r^{-1/3} \log n \right),
\end{align*}
or equivalently, for some constant $C_3$
\begin{align*}
d_G(u,v) &\le 1 + \frac{\|u-v\| (1 + O(r^{-4/3})) + O(r^{-1/3} \log n)}{r}\\
&\le 1 + \frac{\|u-v\|}{r} +  C_3\left( \frac{\|u-v\|}{r^{7/3}} + \frac{\log n}{r^{4/3}}\right)\\
&= 1 + \frac{\|u-v\| +  \kappa}{r}\, .
\end{align*}

Since $d_G(u,v)$ is an integer, and $\frac{\|u-v\| +  \kappa}{r}$ being an integer is a measure zero event, it follows that 
\[
d_G(u,v) \le \left\lceil \frac{\|u-v\| +  \kappa}{r}\right\rceil .
\]
\end{proof}

Let us discuss how we will use Theorems~\ref{thm:short-dis} and~\ref{thm:greedy-routing} to break the $\Omega(r)$ barrier in distance estimation, and thus in reconstruction.
Suppose $r = n^\alpha$ where $0 < \alpha < 1/2$ is a constant. 
Then since $\Vert u-v\Vert  = O(n^{1/2})$, we have from~\eqref{eq:kappa}

\begin{equation}
\label{eq:kappa-alpha}
\frac{\kappa}{r} 
= O\!\left( \max\!\left( n^{\frac{1}{2}-\frac{7}{3}\alpha}, n^{-\frac{4}{3} \alpha} \log n \right) \right) \, ,
\end{equation}
and since $\frac{1}{2}-\frac{7}{3}\alpha > - \frac{4}{3} \alpha$ we have
\begin{equation}\label{eq:kappa-o}
    \kappa = O( n^\beta ) ,
    \quad \text{where} \quad 
    \beta = \frac{1}{2} - \frac{4}{3}\,\alpha \, .
\end{equation}
If $\alpha > 3/14$, then $\beta < \alpha$ and $\kappa = o(r)$. In this case the upper and lower bounds on $d_G(u,v)$ differ by at most $1$, and moreover are equal for most pairs of vertices, making $d_G(u,v)$ a nearly-deterministic function of $\|u-v\|$. Using $\lceil x \rceil \le x+1$ and multiplying through by $r$ gives the bounds
\[
d_G(u,v)r - (r+\kappa) \le \|u-v\| \le d_G(u,v) r \, ,
\]
so that $d_G(u,v)r$ is an estimate of $\|u-v\|$ with error $r+\kappa = (1+o(1))r$. Previous work~\cite{diaz2016relation,diaz2019learning} used this bound to reconstruct the graph with a distortion of $(1+\eps)r$ for arbitrarily small constant $\eps$. 
This gives the performance shown by the dotted line in Figure~\ref{fig:comparison}. 

But in fact $d_G(u,v)r$ is a much more accurate estimate of $\|u-v\|$ for certain pairs of vertices. Namely, if $\|u-v\|$ is just below a multiple of $r$, then rounding up the left and right sides of~\eqref{eq:lenses} doesn't change either very much. We state this with the following corollary. 
\begin{corollary}
\label{cor:kappa}
With $\kappa = \kappa(\|u-v\|)$ defined as in~\eqref{eq:kappa}, suppose that for some $\delta$ and some non-negative integer $t$ we have
\[
tr - (\kappa+\delta) < \|u-v\| < tr - \kappa \, .
\]
Then 
\[
d_G(u,v) r - (\kappa+\delta) \le \|u-v\| \le d_G(u,v) r \, .
\]
\end{corollary}

\begin{proof}
We have
\[
\left\lceil \frac{\|u-v\|+\kappa}{r} \right\rceil \le \frac{\|u-v\|+\kappa+\delta}{r} \, , 
\]
and~\eqref{eq:lenses} then gives the stated result.
\end{proof}

Thus, if $\|u-v\|$ is in one of these intervals, Theorem~\ref{thm:greedy-routing} lets us estimate $\|u-v\|$ from the adjacency matrix with error $\delta+\kappa$ instead of $r+\kappa$. Below we will combine this with the more precise estimate of short-range distances from Lemma~\ref{lem:adjacent-distances} to achieve this error for all pairs $u,v$ of vertices where $v$ is deep, not just those for which $\|u-v\|$ is almost a multiple of $r$. As a result, the error in our distance estimates and the distortion of our reconstruction is $O(r^\beta)$ where $\beta$  decreases from $1$ to $0$ as $\alpha$ increases as shown by the solid line in Figure~\ref{fig:comparison}. Specifically, we obtain a nontrivial result for any $\alpha > 0$ and a more accurate reconstruction than~\cite{diaz2019learning} in the range $\alpha > 3/14$ where their theorem applies. At $\alpha=3/8$, where $\beta=0$ another source of error takes over, leaving us with $O(\sqrt{\log n})$ distortion.

\subsection{Hybrid estimates of long-range distances}\label{sec:hybrid}

In order to combine the long-range estimates of Theorem~\ref{thm:greedy-routing} with the more precise short-range estimates of Lemma~\ref{lem:adjacent-distances}, 
it will be helpful to set up some general machinery.

\begin{definition}
Suppose $V \subset \R^2$. Let $d : V^2 \to [0, \infty)$ and $\eps: \R \to [0, \infty)$ be two functions satisfying,
for all $u, v \in V$, 
\[
d(u,v) - \eps(u,v) \le \|u-v\| \le d(u,v) .
\]
Then we say $d$ is an upper bound on Euclidean distance with error function $\eps$.
\end{definition}

\noindent These error functions will often be bounded by functions of the Euclidean distance, in which case we will write $\eps(\|u-v\|)$ rather than $\eps(u,v)$. In our application, $V$ consists of all vertices in the geometric graph, but we will achieve particularly small $\eps(u,v)$ when $v$ is deep.

\begin{figure}
\begin{center}
\includegraphics[width=3.5in]{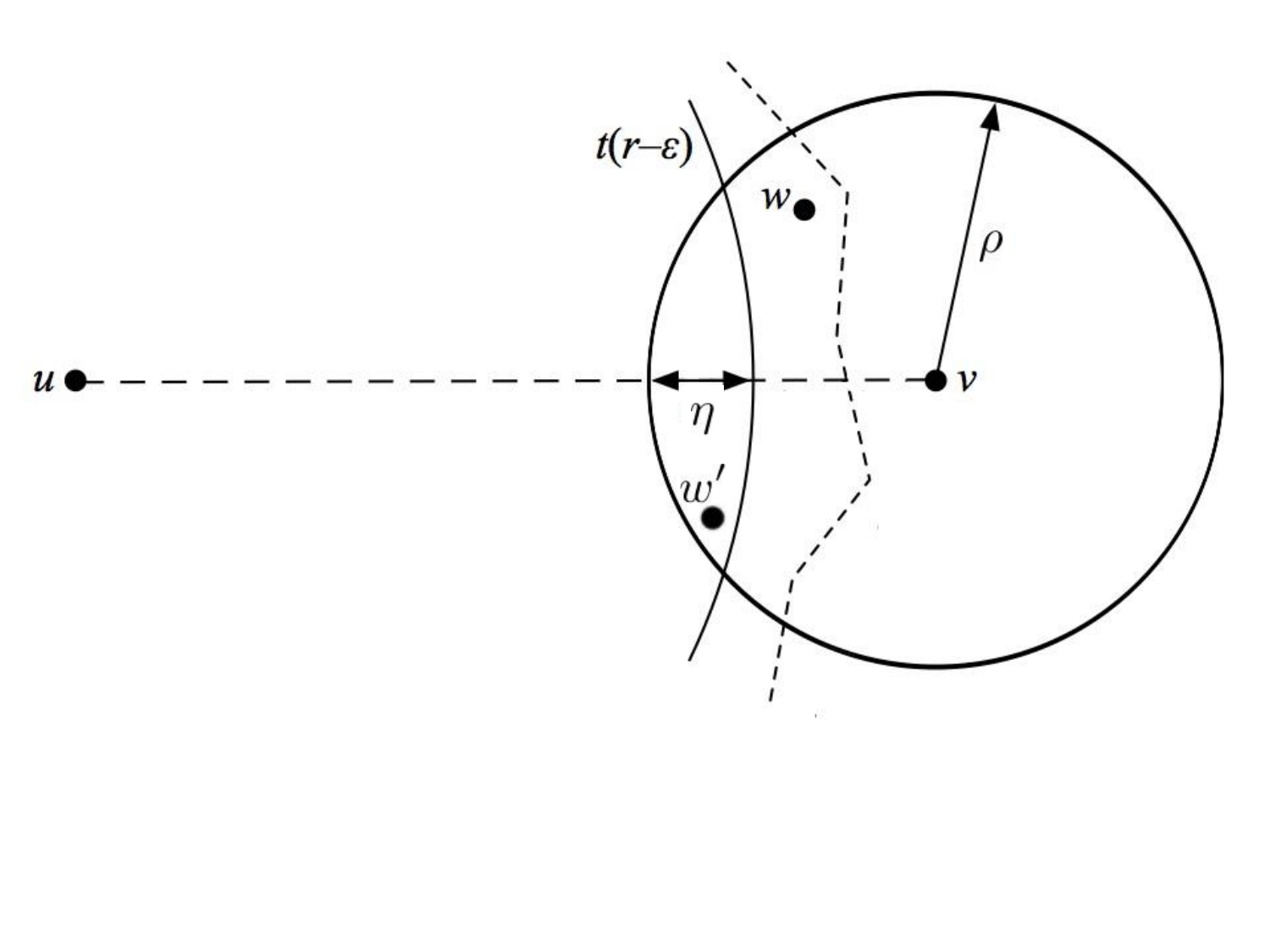}
\end{center}
\caption{For any intermediate point $w$, the hybrid distance estimate $d_1(u,w)+d_2(w,v)$ is an upper bound on $\Vert u-v\Vert $ with error bounded by Lemma~\ref{lem:dhat-error}.}
\label{fig:getting-close}
\end{figure}
A basic tool for combining distance estimates is the following:

\begin{lemma} 
\label{lem:min-dist-estimate}
If $d_1$ and $d_2$ are upper bounds on Euclidean distance with error functions $\eps_1, \eps_2$ respectively, then $\min\{d_1, d_2\}$ is an upper bound on Euclidean distance with error $\min\{\eps_1, \eps_2\}$.
\end{lemma}

\begin{proof}
For any $u,v$ we have $\|u-v\| \le \min\{d_1,d_2\}$. On the other side, assume without loss of generality that $\eps_1 \le \eps_2$. Then $\|u-v\| \ge d_1-\eps_1 \ge \min\{d_1,d_2\}-\eps_1$.
\end{proof}

The next lemma shows another way to combine two upper bounds on $\Vert u-v \Vert$. We choose a vertex $w$ between $u$ and $v$ and use the triangle inequality, using $d_1$ to bound $\Vert u-w \Vert$ and $d_2$ to bound $\Vert w-v \Vert$ (see Fig~\ref{fig:getting-close}). Finally, we minimize over all intermediate vertices $w$. This hybrid is especially useful when, as with our long-range and short-range estimates, $d_1$ and $d_2$ have different ranges of $\Vert u-v\Vert $ in which they achieve small error.
\begin{figure}
\begin{center}
\includegraphics[width=6cm]{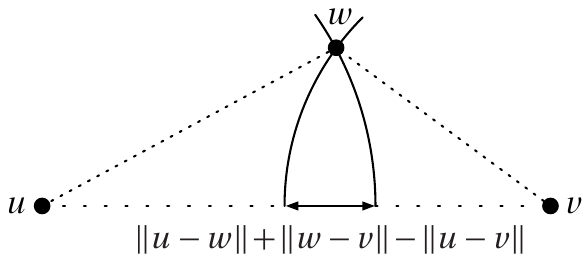}
\end{center}
\caption{Figure for the proof in Lemma~\ref{lem:dhat-error}.}
\label{fig:uvwd}
\end{figure}

\begin{lemma} 
\label{lem:dhat-error}
Suppose $d_1$ and $d_2$ are upper bounds on Euclidean distance with error functions $\eps_1$ and $\eps_2$ respectively. Define the hybrid distance estimate $\dhat$ by
\[
\dhat = \min_w \big( d_1(u,w) + d_2(w,v) \big) \, . 
\]
Then $\dhat$ is an upper bound on Euclidean distance with error 
\[
\epshat(u,v) \le \min_w \Big[ \eps_1(u,w) + \eps_2(w,v) + \Vert u-w\Vert  + \Vert w-v\Vert  - \Vert u-v\Vert  \Big] \, .
\]
\end{lemma}
\begin{proof}
Fix a pair of vertices $u, v$. For any $w$, by the triangle inequality we have
\[
\Vert u-v\Vert  \le \Vert u-w\Vert  + \Vert v-w\Vert  \le d_1(u,w) + d_2(w,v) \, ,
\]
so $\dhat$ is an upper bound on Euclidean distance. On the other hand, as shown in Figure~\ref{fig:uvwd} we have
\begin{align*}
\Vert u-v\Vert 
&= \Vert u-w| + \Vert v-w\Vert  - \big( \Vert u-w\Vert +\Vert w-v\Vert -\Vert u-v\Vert  \big) \\
&\ge d_1(u,w) + d_2(v,w) - \big( \eps_1(u,w) + \eps_2(w,v) + \Vert u-w\Vert +\Vert w-v\Vert -\Vert u-v\Vert  \big) \, .
\end{align*}
Using Lemma~\ref{lem:min-dist-estimate} to minimize the error over $w$ completes the proof. 
\end{proof}

Next, we will use the fact that if a lens is large enough to contain at least one point $w$ with high probability, this yields an upper bound on the minimum in Lemma~\ref{lem:dhat-error}.
\begin{lemma}\label{lem:hybrid-lens}
Suppose $G$ is a random geometric graph, and that with high probability, $d_1$ and $d_2$ are upper bounds on Euclidean distance with errors $\eps_1(u,v)=\eps_1(\|u-v\|)$ 
and $\eps_2(u,v)=\eps_2(\|u-v\|)$. Define the hybrid distance $\dhat$ as in Lemma~\ref{lem:dhat-error}. Then there is a constant $C$ such that, with high probability, $\dhat$ is also an upper bound on Euclidean distance, with error $\epshat(u,v) = \epshat(\|u-v\|)$ where
\begin{equation}
\label{eq:min-over-x}
\epshat(\Vert u-v\Vert ) 
\le \min_{0 < x < \Vert u-v\Vert } 
\max_{0 \le \delta_1, \delta_2 \le \delta(x)} 
\big[ \eps_1(x + \delta_1) + \eps_2(\Vert u-v\Vert  - x - \delta_2) + \delta(x) \big] ,
\end{equation}
where 
\begin{equation}
\label{eq:delta-big-enough}
\delta(x) = C (\log n)^{2/3} \left( \min\{x,\Vert u-v\Vert -x\} \right)^{-1/3} \, . 
\end{equation}
\end{lemma}

\begin{figure}
\begin{center}
\includegraphics[width=6.5cm]{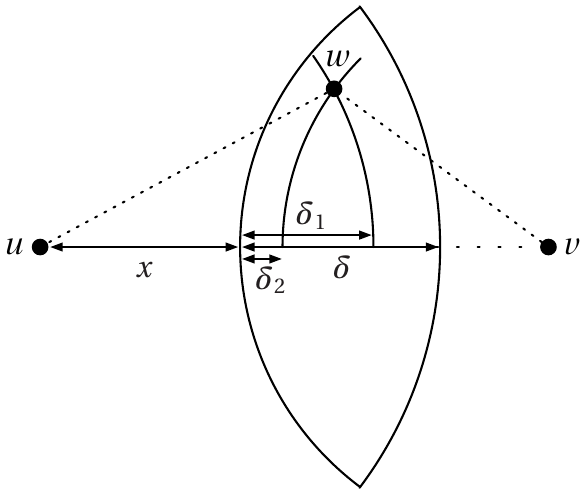}
\end{center}
\caption{The lens $L(x)$ of Lemma~\ref{lem:hybrid-lens}. If $\delta$ is large enough, this lens is nonempty with high probability, in which case we can use any point $w$ in it as an intermediate point for Lemma~\ref{lem:dhat-error}.}
\label{fig:hybrid-lens}
\end{figure}
\begin{proof}
Fix a pair $u,v$, and consider the lens $L(x)$ of width $\delta$ consisting of the intersection of the balls centered at $u$ and $v$ of radius $r_1 = x+\delta$ and $r_2 = \|u-v\| - x$ respectively as shown in Figure~\ref{fig:hybrid-lens}. Lemma~\ref{lem:lens-width-area} and $x \le r_1$ shows that the area of $L(x)$ is proportional to $C^{3/2} \log n$.

While $x$ takes an infinite number of values, for all $x$ we have 
\[
\delta(x) \ge \delta_{\min} = \Omega\!\left( (\log n)^{3/2} \|u-v\|^{-1/3} \right) \, . 
\]
Now consider $2 \|u-v\| / \delta_{\min}$ lenses of width $\delta_{\min}/2$, where $x$ is an integer multiple of $\delta_{\min}/2$. For any $x$, $L(x)$ contains one of these smaller lenses, so if they are all nonempty, so is $L(x)$ for all $x$. Each of these small lenses also has area $\Omega(\log n)$, and there are $O(\|u-v\|/\delta_{\min}) = O(\|u-v\|^{4/3})=O(n^{2/3})$ of them. Thus if we set the constant $C$ large enough for the area of the smaller lenses to be $3 \log n$, say, the probability that any of them are empty, for any pair $u,v$, is $o(1)$. 

Hence $L$ contains at least one vertex $w$ with high probability. Applying Lemma~\ref{lem:dhat-error} to $w$, writing $\|u-w\| = x + \delta_1$ and $\|w-v\| = \|u-v\|-x-\delta_2$ as shown in Figure~\ref{fig:hybrid-lens}, and pessimistically maximizing $\eps_1(u,w)$ and $\eps_2(w,v)$ over the lens yields the desired upper bound on the error of $\dhat$. 
\end{proof}

We use the previous lemma to break the $\Omega(r)$ barrier for the error in estimating Euclidean distances in $G \in \RG$.

Now assume $v$ is deep, and define $d_1$ and $d_2$ as follows:
\begin{align}\label{eq:dmind1d2}
d_1(u,v) &= r d_G(u,v) , \nonumber \\
d_2(u,v) &= \begin{cases} 
\dtilde(u,v) + C_2 \sqrt{\log n} & \text{if $d_G(u,v) \le 2$} , \nonumber \\
+\infty & \text{otherwise} \, ,
\end{cases}
\end{align}
where $\dtilde$ is defined as in equation~\eqref{eq:dtilde} and $C_2$ is the constant in equation~\eqref{eq:dtilde-good}. Thus $d_1$ is the upper bound of Corollary~\ref{cor:kappa}, and $d_2$ is the precise short-range estimate $\dtilde$ of Lemma~\ref{lem:adjacent-distances} with a small increment to make it an upper bound on Euclidean distance with high probability.

\begin{remark}\label{rem:dijkstra}
Given this choice of $d_1$ and $d_2$, the hybrid estimate $\dhat(u,v)$ is the graph distance from $u$ to $v$ in a weighted graph $G_v$ where each edge $(w,v)$ with $d_G(w,v) \le 2$ has weight $\dtilde(w,v) + C_2 \sqrt{\log n}$ and all other edges have weight $r$. Thus, for any fixed $v$, we can compute $\dhat(u,v)$ for all $u$, using a modified breadth-first search algorithm, in $O(n + m + \Delta \log \Delta)$ steps, where $\Delta$ is the degree of vertex $v$.  The only modification to the standard queue-based implementation of BFS, is to sort the level-one nodes by increasing order of $\dtilde$, before placing them in the queue.  Since all other edges have the same weight, $r$, which exceeds the difference between any two vertices in the queue, 
it follows by induction that all vertices will be discovered in increasing order of $\dhat$.
\end{remark}

We bound the errors $\eps_1$ and $\eps_2$ as follows. As discussed after Theorem~\ref{thm:greedy-routing}, for most values of $\Vert u-v\Vert $ we have $\eps_1(\Vert u-v\Vert ) = \Theta(r)$. However, we will choose the lens in Lemma~\ref{lem:hybrid-lens} so that $\Vert u-w\Vert $ is almost a multiple of $r$, in which case Corollary~\ref{cor:kappa} shows that $\eps_1(\Vert u-w\Vert )$ is much smaller.  

For $\eps_2$,  Lemma~\ref{lem:adjacent-distances} implies that, for some absolute constant $C_4$, with high probability
\begin{equation}\label{eq:eps2}
\eps_2(\|u-v\|) \le \begin{cases}  
C_4 \sqrt{\log n} & \text{if $\|u-v\| \le 2r - C_4 r^{-1/3} \log n$} ,  \\
+\infty & \text{otherwise} \, .
\end{cases} 
\end{equation}
Here we used equation~\eqref{eq:kappa} and the upper bound of Theorem~\ref{thm:greedy-routing} to show that with high probability $d_G(u,v) \le 2$ whenever
\[
\|u-v\| 
\le 2r-\kappa(2r) 
= 2r - C_3 r^{-1/3} (2 + \log n) \, ,
\]
and we set $C_4 > \max\{C_3, 2C_2 \}$. 

Having gathered these facts, we will apply Lemma~\ref{lem:hybrid-lens} to $d_1$ and $d_2$ with a judicious choice of lens $L(x)$. First note that, since $d_2(w,v)=+\infty$ if $d_G(w,v) > 2$, we can write the hybrid distance estimate as,
\begin{equation}\label{eq:dhat}
    \dhat(u,v) = \min_{w: d_G(w,v) \le 2} d_1(u,w)+d_2(w,v)
\end{equation}

\begin{theorem}
\label{thm:dhatmind1d2}
Let $r=n^\alpha$ for a constant $0 < \alpha < 1/2$. For all pairs $u, v$ where $v$ is deep, define $\dhat(u,v)$ as in equation~\eqref{eq:dhat}. Then with high probability, $\dhat$ is an upper bound on the Euclidean distance $\Vert u-v\Vert $ with error 
\begin{equation}
\label{eq:epshat}
\epshat(u,v) \le C' \begin{cases}
n^{\frac{1}{2} - \frac{4}{3} \alpha} & \alpha < 3/8, \\
\sqrt{\log n} & 3/8 \le \alpha < 1/2 ,
\end{cases}
\end{equation}
for some absolute constant $C'$. That is,  
\[
\dhat(u,v) - \epshat(u,v) \le \Vert u-v\Vert  \le \dhat(u,v) \, . 
\]
\end{theorem}

\begin{proof}
Below we will bound the error of $\dhat$ as $\epshat = \kappa + O(\sqrt{\log n})$. Recall that if $0 < \alpha < 1/2$, Equation~\eqref{eq:kappa} gives 
\[
\kappa = (1+o(1)) C_3 n^\beta 
\quad \text{where} \quad 
\beta = \frac{1}{2} - \frac{4}{3} \,\alpha \, . 
\]
If $\alpha < 3/8$ then $\beta > 0$ and $\epshat = (1+o(1)) \kappa$. If $\alpha \ge 3/8$ then $\beta \le 0$, $\kappa=O(1)$, and $\epshat = O(\sqrt{\log n})$.

We can upper bound equation~\eqref{eq:min-over-x} with any choice of $x$ between $0$ and $\Vert u-v\Vert $, i.e.\ any lens between $u$ and $v$ that \whp\ contains an intermediate point $w$. We choose this lens as follows. First let $t$ be the largest integer such that $\Vert u-v\Vert  > tr + r/2$. We assume without loss of generality that $t \ge 1$, since if $t=0$ we have $\Vert u-v\Vert  < 3r/2$ and equation~\eqref{eq:eps2} gives $\epshat \le \eps \le \eps_2 = O(\sqrt{\log n})$. 

Let $C$ be the constant required by Lemma~\ref{lem:hybrid-lens}. If we define
\begin{equation}
\label{eq:hybrid-delta}
\delta = C (\log n)^{2/3} (r/2)^{-1/3} 
\end{equation}
and set $x = tr-(\kappa+\delta)$, we have $\min\{x,\Vert u-v\Vert -x\} \ge r/2$, and $\delta$ satisfies the condition~\eqref{eq:delta-big-enough}. That is, the lens $L(x)$ of width $\delta$ shown in Figure~\ref{fig:hybrid-lens} contains at least one point $w$ with high probability. 

Now we can bound the errors in $d_1(u,w)$ and $d_2(w,v)$. First,\ we have chosen $L(x)$ so that $\Vert u-w\Vert $ is almost an integer, i.e.\ $tr - (\kappa+\delta) < \Vert u-w\Vert  < tr - \kappa$. Then Corollary~\ref{cor:kappa} tells us that $d_1(u,w)$ has error bounded by  $\eps_1 \le \kappa+\delta$. Thus, for all $0 \le \delta_1 \le \delta$, the first term of  equation~\eqref{eq:min-over-x} is bounded by $\eps_1(x + \delta_1) \le \kappa + \delta$. 

Second, we have $\Vert v-w\Vert  \le \Vert u-v\Vert  - x \le 3r/2 + \kappa + \delta$, so \whp\ $d_G(w,v) \le 2$ and equation~\eqref{eq:eps2} gives $\eps_2(w,v) \le C_4 \sqrt{\log n}$. Thus, for all $0 \le \delta_2 \le \delta$, the second term of equation~\eqref{eq:min-over-x} is bounded by $\eps_2(\Vert u-v\Vert  - x - \delta_2) \le C_4 \sqrt{\log n}$.

Combining all of this, Lemma~\ref{lem:hybrid-lens} tells us that with high probability
\[
\epshat(\Vert u-v\Vert ) \le \kappa + \delta + C_4 \sqrt{\log n} + \delta 
= \kappa + C_4 \sqrt{\log n} + o(1) \, .
\]
Setting $C' > \max\{C_3,C_4\}$ completes the proof.
\end{proof}
\section{The Reconstruction Algorithm}\label{sec:reconstruction}

In this section we use our distance estimates to reconstruct the positions of the points up to a symmetry of the square. Our global strategy is similar to~\cite{diaz2019learning}: we first fix a small number of ``landmark'' vertices $v$ whose positions can be estimated accurately up to a symmetry of the plane. Then for each vertex $u$ we use the estimated distances $\dhat(u,v)$ to reconstruct $u$'s position by triangulation. In~\cite{diaz2019learning}, the landmarks are vertices close to the corners of the square. Here they will instead be a set of three deep vertices that are far from collinear, forming a triangle which is acute and sufficiently large.  
\begin{lemma}\label{lem:xyz}
Let $x,y,z,u$ be four points in the plane.
Suppose $x,y,z$ form an acute triangle with minimum side length at least $\ell$. Then, if we know the positions of $x,y,z$ with error at most $\eta$, and we have upper bounds $\dhat(u,v)$ on the Euclidean distances $\|u-v\|$ for all $v \in \{x,y,z\}$ with error $\epshat$, and all of these distances are at most $D$, we can determine the position of $u$ relative to $x,y,z$ with error at most 
\begin{equation}\label{eq:xyz}
C_5 \frac{D (\epshat + \eta)}{\ell} ,
\end{equation}
for an absolute constant $C_5$. 
\end{lemma}
%
\begin{proof}
First, let us assume fixed positions for $x,y,z$ within $\eta$ of their estimated positions (which we can always do so that they form an acute triangle). By the triangle inequality, this changes the distances $\|u-v\|$ for $v \in \{x,y,z\}$ by at most $\pm\eta$. Thus $u$ is in the intersection $U$ of three annuli,
\begin{equation}
\label{eq:annuli}
U = \bigcap_{v \in \{x,y,z\}}
B(v, \dhat(u,v) + \eta) \setminus
B(v, \dhat(u,v) - \eta - \epshat) \, .
\end{equation}
Any point $u'$ in $U$ gives an approximation of $u$'s position with error at most the Euclidean diameter of $U$, namely $\max_{u,u' \in U} \|u-u'\|$. We will show this diameter is bounded by equation~\eqref{eq:xyz}.

We use some basic vector algebra. Let $\eps' = \epshat + 2\eta \le 2(\epshat+\eta)$. For any $u,u' \in U$ we have, for all $v \in \{x,y,z\}$,
\[
-\eps' \le \Vert  u - v \Vert  - \Vert  u' - v\Vert  \le \eps' \, .
\]
Since the triangle $x,y,z$ is acute, at least one of its sides makes an angle $\varphi$ with the vector $u - u'$ where $0 \le \varphi \le \pi/4$.  Taking this side to be $(x,y)$ we have, without loss of generality,
\[
(y - x) \cdot (u - u') = \Vert  y - x \Vert  \; \Vert u - u'\Vert  \cos \varphi \ge \Vert  y - x \Vert  \; \Vert u - u'\Vert  \sqrt{\frac12} .
\]
Next, we rewrite this dot product as follows,
\begin{align*}
2 (y-x) \cdot (u - u') 
&= \Vert x - u\Vert ^2 - \Vert x - u'\Vert ^2 - \Vert  y - u\Vert ^2 + \Vert y - u'\Vert ^2 \\
&= (\Vert x-u\Vert  - \Vert x-u'\Vert )(\Vert x-u\Vert +\Vert x-u'\Vert ) \\
& \ \ \ \  - (\Vert y-u\Vert  - \Vert y-u'\Vert )(\Vert y-u\Vert +\Vert y-u'\Vert ) \\
& \le  \eps' (\Vert x-u\Vert +\Vert x-u'\Vert  + \Vert y-u\Vert +\Vert y-u'\Vert ) \, ,
\end{align*}
where the first line is a classical polarization identity. Putting these together, we have
\begin{align*}
\Vert u - u'\Vert  
&\le \frac{\sqrt 2}{\Vert y-x\Vert } \,\eps' 
\big( \Vert x-u\Vert +\Vert x-u'\Vert  + \Vert y-u\Vert +\Vert y-u'\Vert  \big) 
 \le \frac{4 \sqrt{2} R \eps'}{\ell} \, ,
\end{align*}
completing the proof with $C_5 = 8\sqrt{2}$.
\end{proof}
\begin{remark}
In our application, $D$ is at most the diameter $\sqrt{2n}$ of the square, $\ell=\Omega(\sqrt{n})$, 
and $\eta=O(\epshat)$. Thus we can reconstruct $u$'s position relative to $x,y,z$ with error 
$O(\epshat)$.
\end{remark}
\begin{theorem}\label{thm:main}
Let $r=n^\alpha$ for a constant $0 < \alpha < 1/2$. There is an algorithm with running time $O(n^2)$ that with high probability reconstructs the vertex positions of a random geometric graph, modulo symmetries of the square, with distortion an absolute constant times times $\epshat$ as defined in~\eqref{eq:epshat}, i.e.\ 
\begin{equation*}
d^* = C'' \begin{cases}
n^{\frac{1}{2} - \frac{4}{3} \alpha} & \mbox{ if } \alpha < 3/8 , \\
\sqrt{\log n} & \mbox{ if } 3/8 \le \alpha < 1/2 ,
\end{cases}
\end{equation*}
for some absolute constant $C''$.
\end{theorem}
%
\begin{proof} 

We use the fact, proved in~\cite{diaz2019learning}, that \whp\ the true positions of the lowest-degree vertices are within $\sqrt{\log n}$ of the corners of the square, and we can find those lowest degree vertices  in $O(n^2)$. Call these vertices $a,b,c,d$. 

We construct a good triple, i.e., deep vertices $x,y,z$ that form an acute triangle with a minimum side length at least $\ell=0.1 \sqrt{n}$. Recall that by the bounds of Theorem~\ref{thm:greedy-routing}, with high probability any triple of deep vertices whose graph distances are in the range $[0.1 \sqrt{n}/r, 0.14 \sqrt{n}/r]$ qualifies.

There are many ways to find such a triple. One is to find a vertex $x$ deep inside the square, e.g., with a graph distance at least $0.65 \sqrt{n}/r$ from $a,b,c,d$. We can then take $y$ to be any vertex such that $d_G(x,y)$ is in the interval $[0.1 \sqrt{n}/r, 0.14 \sqrt{n}/r]$, and $z$ to be any vertex such that $d_G(x,z)$ and $d_G(y,z)$ are both in this interval. At each stage of this process such a vertex exists with very high probability, and by Theorem~\ref{thm:greedy-routing} and Lemma~\ref{lem:lots-of-deep} all of them are deep. Finding $x$ takes $O(n)$ time by breadth-first search from $a,b,c,d$, and finding $y$ and $z$ similarly takes $O(n)$ time each. (A randomized algorithm can simple sampling triples uniformly at random: since a constant fraction of triples qualify, this succeeds with high probability within $O(\log n)$ tries.)

Theorem~\ref{thm:dhatmind1d2} gives us estimated side lengths $\dhat(x,y), \dhat(y,z), \dhat(x,z)$ that have error $\epshat$. This lets us estimate the positions of $x,y,z$ modulo an isometry of $\R^2$, i.e.\ a translation, rotation, or reflection of the plane. Equivalently, it lets us construct a triangle $x, y, z$ which is congruent to their true positions up to distortion~$\eta = O(\epshat)$. Then we use Lemma~\ref{lem:xyz} to reconstruct the position of each vertex $u$ relative to this triangle. Given $\dhat(u,v)$ for all $v \in \{x,y,z\}$, in~\eqref{eq:xyz} we have $D \le \sqrt{2n}$ and $\ell \ge 0.1 \sqrt{n}$, 
Lemma~\ref{lem:xyz} gives us $u$'s position relative to $x,y,z$ with error $10 \sqrt{2} C_5 (\eps+\eta) = O(\epshat)$. 

This gives us a reconstruction up to an isometry as shown in Figure~\ref{fig:rotated}. Finally, we rotate and translate this reconstruction to the square $[0,\sqrt{n}]^2$. It is easy to compute an isometry that sends $\{a,b,c,d\}$ to the corners $\SRn$ with error $d^* = O(\sqrt{\log n})$: for instance, translate $a$ to $(0,0)$ and then rotate one of the closer corners to $(0,\sqrt{n})$. This gives a reconstruction which, up to a rotation or reflection of the square, has distortion $d^* = O(\epshat+\sqrt{\log n})=O(\epshat)$.

Step~1 can be done by breadth-first search, first from $a,b,c,d$ and then from $x$ and $y$, and thus takes $O(n)$ time. 
The bulk of the running time comes from computing the graph distances. Notice we have a few real-valued and geometric calculations to do. 
In order to compute $\dhat(u,v)$ for all $u$ and all $v \in \{x,y,z\}$, we need $\dtilde(w,v)$ for all $v \in \{x,y,z\}$ and all $w$ with $d_G(w,v) \le 2$. There are $O(n^2)$ such $w$. Inverting the function $F$ in equation~\eqref{eq:dtilde} is a matter of arithmetic; we need to do this to $O(\log r) \le O(\log n)$ bits of precision, which can be done in $\polylog(n)$ time. Thus computing these $\dtilde$ can be done in time $O(r^2 \,\polylog(n))$, which is $o(n)$ since $\alpha < 1/2$.

Once we have the distance estimates $\dhat(u,v)$, finding estimated positions $u$ in each region $U$ defined in equation~\eqref{eq:annuli} is a geometric calculation which can be done to $O(\log n)$ bits of precision in $\polylog(n)$ time as in~\cite{diaz2019learning}. The same is true of computing the angle by which we need to rotate the reconstruction to $[0,\sqrt{n}]^2$ after translating one corner to the origin. 

Since the typical degree in the graph is $\pi r^2 = O(n^{2\alpha})$ where $\alpha < 1/2$, and since Dijkstra's algorithm in a graph with $n$ vertices and $m$ edges runs in time $O(m+n \log n)$, the running time is \whp\ $O(n^{2\alpha}+1)=o(n^2)$. Thus the total running time is dominated by that of Dijkstra's algorithm, which we generously bound as $O(n^2)$.
\end{proof}

\begin{figure}
    \centering
    \includegraphics[width=9cm]{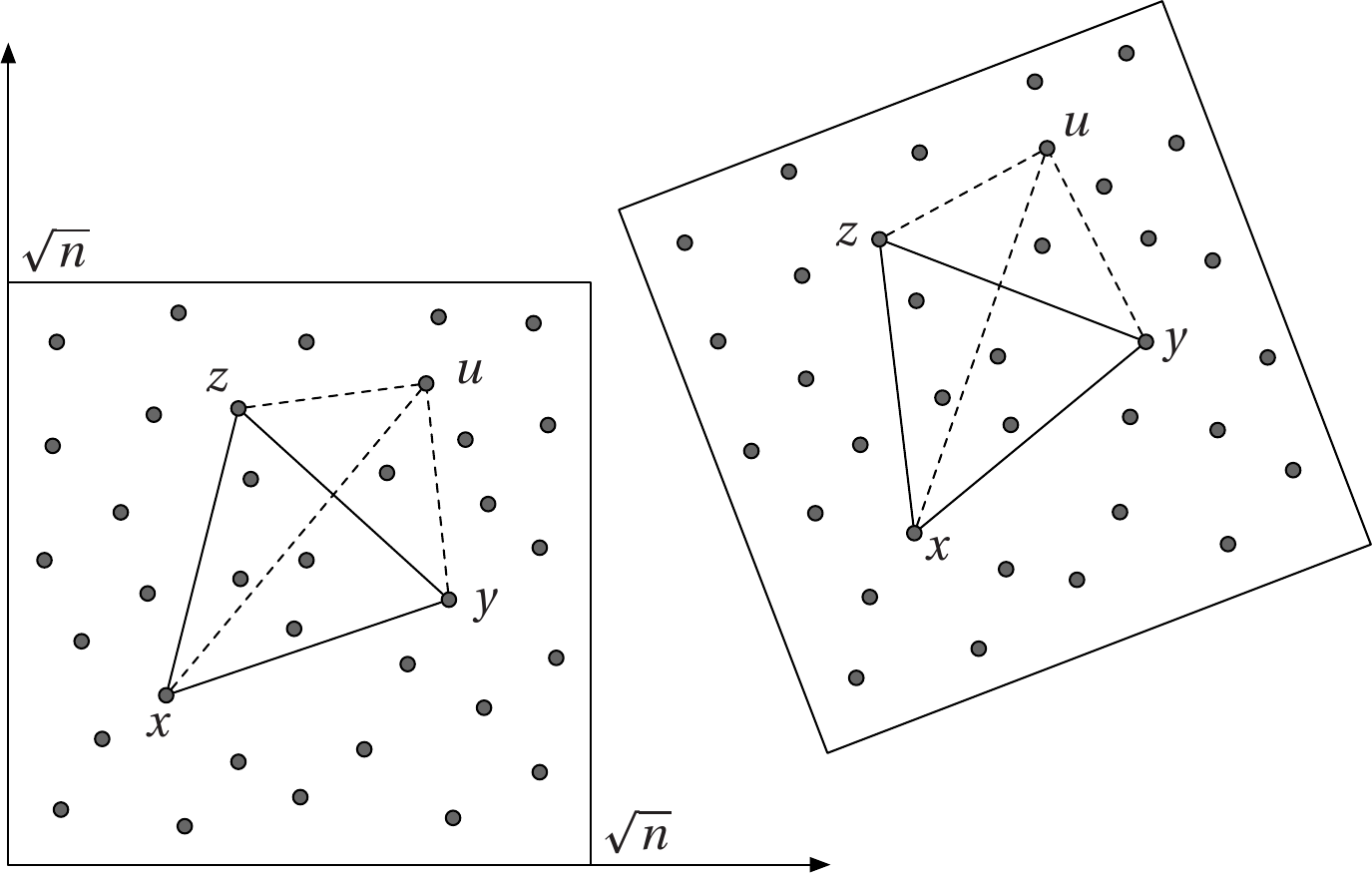}
    \caption{Our reconstruction is built around a triangle $x,y,z$ of deep vertices. It may be translated, rotated, or reflected in $\R^2$ by an isometry, but it can then be shifted to the square $[0,\sqrt{n}]^2$. Then it will be a good reconstruction up to a rotation or reflection of the square.}
    \label{fig:rotated}
\end{figure}

\begin{remark}
After we have reconstructed the positions of all vertices, i.e. their coordinates, we can extend our distance estimates to all the pairs $u,v$ of vertices that are not deep, simply by estimating $\Vert u-v \Vert$ as the distance between the reconstructed positions of $u$ and $v$.
\end{remark}

\section{Extensions to Other Domains}\label{sec:extensions}
Our results can be generalized from the square to a number of alternative domains for random geometric graphs, including to higher-dimensional Euclidean space and to some curved manifolds. Here we show that our results extend to the $m$-dimensional hypercube and the $3$-dimensional sphere, solving an open problem posed in~\cite{diaz2019learning}.
\subsection{Reconstruction in higher-dimensional Euclidean space}
The simplest generalization is where the underlying domain is $[0,n^{1/m}]^m \subset \R^m $, 
i.e., an $m$-dimensional hypercube with volume $n$. where $m$ is fixed, i.e., it does not vary with $n$. Let the radius of \jdc{$G$} be $r$. Then, $n$ points are selected at random from the hypercube and pairs of  them are adjacent if they are within Euclidean distance $r$. As before, the goal is to reconstruct the points' positions based on the adjacency matrix of the graph. 

The following lemma generalizes Lemm~\ref{lem:lens-width-area}, and relates the width of a lens defined as the intersection of two balls of radius $r$ to its $m$-dimensional volume. 

\begin{lemma} \label{lem:lens-width-volume}
Let $B_1$, $B_2$ be two overlapping Euclidean balls of radius $r_1 \le r_2$, respectively, in $\R^m$.  
Consider the lens $L = B_1 \cap B_2$.  
Let $\eps$ denote the width of $L$, that is, 
\[
\eps = \min\{r_1 + r_2 - d, 2r_1 \},
\]
where $d$ is the distance between the two centers.
Then the volume $V$ of $L$ satisfies
\[
V = \Theta\left( \sqrt{\eps^{m+1} r_1^{m-1}} \right),
\]
where the hidden constant depends only on $m$.
\end{lemma}
\begin{proof}
Assume $\eps \le r_1$, since otherwise the upper bound follows trivially from the volume of the ball.
Our lower bound for the case $\eps > r_1$ will follow from the case $\eps \le r_1$.

Let $C$ be the cylinder circumscribed around $L$, with axis of symmetry the line joining the centers of $B_1$ and $B_2$.  See Figure~\ref{fig:lens-area}. Then $C$ has height $\eps$, and radius $\sqrt{r_1^2 - (r_1-\eps_1)^2}
= \sqrt{r_2^2 - (r_2-\eps_2)^2} = O(\sqrt{r_1 \eps_1}) = O(\sqrt{r_2 \eps_2})$, where $\eps_1, \eps_2$ are the heights of the two spherical segments comprising $L$.  Since $\eps_1 \le \eps$, 
we find that the volume of this cylinder is $O(\eps (\sqrt{r_1 \eps})^{m-1})$, as desired.

\begin{figure}
\begin{center}
\includegraphics[width=4in]{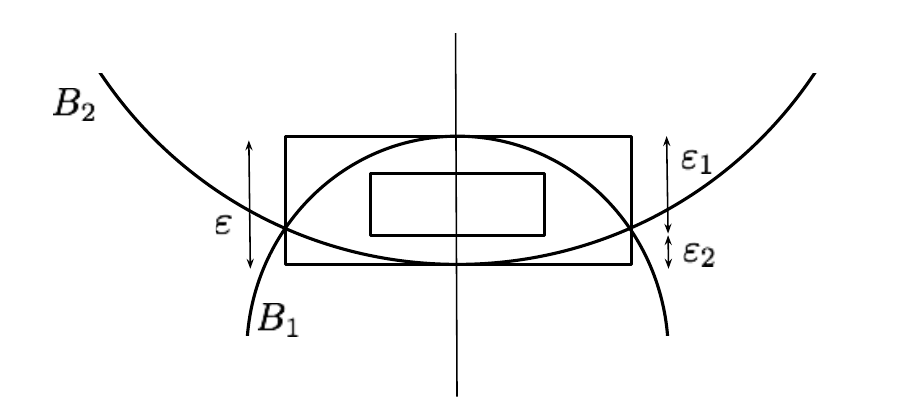}
\end{center}
\caption{The figure illustrates how we can bound the area of a lens using inscribed and circumscribed cylinders.}
\label{fig:lens-area}
\end{figure}

Similarly, if we scale down the height and radius of $C$ by a factor of $2$, and translate it appropriately,
we can place this re-scaled cylinder completely within the lens, $L$.  Since this changes the volume of the cylinder by a factor of $2^{-m}$, which depends only on $m$, this establishes the desired lower bound.
\end{proof}

Given this relation between the width and volume of the lens, analogously to the work in Section~\ref{sec:dist} , we can compute both short and long range estimates of the distance, and then combine them into a hybrid estimate. The error in the hybrid estimate is given by the following theorem.
\begin{theorem}
\label{thm:highdim-dhatmind1d2}
Let $r=n^\alpha$ for a constant $0 < \alpha < 1/m$. For all pairs $u, v$ where $v$ is deep, define $\dhat(u,v)$ be the hybrid estimate of the distance.  Then with high probability, $\dhat$ is an upper bound on the Euclidean distance $\|u-v\|$ with error 
\begin{equation}
\label{eq:hd-epshat}
\epshat(u,v) \le C_m \begin{cases}
n^{\frac{1}{m} - \frac{2m}{m+1} \alpha} & \alpha < \frac{m+1}{2m^2} , \\
\sqrt{\log n} &  \frac{m+1}{2m^2}\le \alpha < \frac{1}{m} ,
\end{cases}
\end{equation}
for some dimension-dependent constant $C_m$.
\end{theorem}

\noindent We omit the details of the proof except to comment that it closely follows the steps in Section~\ref{sec:dist}. 

In order to use the hybrid estimates for reconstruction, we need to find an appropriate number of deep landmarks. By standard linear algebra, it suffices to have $m+1$ landmarks that form a non-degenerate simplex. As in the proof of Theorem~\ref{thm:main}, we can find an approximately equilateral $m$-simplex in quadratic time, directly from the adjacency matrix of the graph. Finally, once we have reconstructed the relative positions of all the points, it is easy to compute an isometry that shifts the reconstructed hypercube to the origin. 

In order to use the hybrid estimates for reconstruction, we need to find an appropriate number of deep landmarks. By standard linear algebra, it suffices to have $m+1$ landmarks that form a non-degenerate simplex. As in the proof of Theorem~\ref{thm:main}, we can find an approximately equilateral $m$-simplex in quadratic time, directly from the adjacency matrix of the graph. Finally, once we have reconstructed the relative positions of all the points, it is easy to compute an isometry that shifts the reconstructed hypercube to the origin. 

Putting it all together we have the following theorem for reconstruction of RGGs in $[0, n^{1/m}]^m$. 
\begin{theorem}\label{thm:highdim}
Let $r=n^\alpha$ for a constant $0 < \alpha < 1/m$. There is an algorithm with running time $O(n^\omega \log n)$, where $\omega < 2.373$ is the matrix multiplication exponent~\cite{AlmanW21}, that with high probability reconstructs the vertex positions of a random geometric graph, modulo symmetries of the hypercube, with distortion 
\begin{equation*}
d^* \le C_m \begin{cases}
n^{\frac{1}{m} - \frac{2m}{m+1} \alpha} & \mbox{ for } \alpha < \frac{m+1}{2m^2}, \\
\sqrt{\log n} &  \mbox{ for } \frac{m+1}{2m^2}\le \alpha < \frac{1}{m} ,
\end{cases}
\end{equation*}
for some dimension-dependent constant $C_m$.
\end{theorem}
\subsection{Reconstruction on the sphere}

Finally, we argue that our algorithm also works on some curved manifolds and submanifolds where the geometric graph is defined in terms of geodesic distance. In particular we claim this for the $m$-dimensional spherical (hyper)surface $S_m$ of a ball in $\R^{m+1}$. Here we sketch the proof for the two-dimensional surface of a sphere in $\R^3$. Note that the distortion is now defined by minimizing over the sphere's continuous symmetry group, i.e., over all rotations and reflections of the sphere. 

In previous work, the authors of~\cite{Bubeck16} gave a procedure to distinguish random geometric graphs on $S^m$ from Erd\H{o}s-R\'enyi random graphs. In addition, \cite{araya2019} gave a spectral method for reconstructing random graphs generated by a sparsified graphon model on the sphere, but this does not include the geodesic disk model we study here.

We define random geometric graphs on the sphere as follows. We scale the sphere so that its surface area is $n$, setting its radius to $R=\sqrt{n/(4\pi)}$. We scatter $n$ points uniformly at random on it, or generate them with a Poisson point process with intensity $1$. In either model the expected number of points in a patch of surface is equal to its area. We define the graph as $(u,v) \in E$ if and only if 
$\|u-v\|_g < r$ where $\|u-v\|_g$ is the geodesic distance, i.e., the length of the shorter arc of a great circle that connects $u$ and $v$. If we associate each point $u$ with a unit vector $\vec{u} \in \R^3$ that points toward it from the center of the sphere, $\|u-v\|_g$ is $R$ times the angle between $\vec{u}$ and 
$\vec{v}$. 
\begin{theorem}
\label{thm:sphere}
Let $r=n^\alpha$ for a constant $0 < \alpha < 1/2$. There is an algorithm with running time $O(n^\omega \log n)$ that with high probability reconstructs the vertex positions of a random geometric graph, modulo a rotation or reflection of the sphere, with distortion an absolute constant times $n^{\frac{1}{2} - \frac{4}{3} \alpha}$ if $\alpha < 3/8$ and $\sqrt{\log n}$ if $\alpha \ge 3/8$.
\end{theorem}
\begin{proof}[Proof Sketch]
First note that objects of size $o(R)$ act nearly as they do in the flat plane. In particular, the neighborhood $B(u,r)$ of each point $u$ is a spherical cap with angular radius $r/R$. The area of this cap is 
\[
2 \pi R^2 \left( 1-\cos (r/R) \right) 
= \pi r^2 \left( 1 - O(r^2/R^2) \right)
= \pi r^2 \left( 1 - O(r^2/n) \right) \, .
\]
If $r=n^\alpha$ where $\alpha < 1/2$, the area of this neighborhood is $1-o(1)$ times the area $\pi r^2$ of the corresponding ball in the flat Euclidean plane. Similarly, Lemma~\ref{lem:lens-width-area} for the area of a lens-shaped intersection between two caps holds as long as $r_1,r_2 = o(R)$.

The area of the lune $B(u,r) \setminus B(v,r)$ is a continuous function $F(\|u-v\|)$ analogous to equation~\eqref{eq:F}, and if we define the short-range distance estimate $\dtilde$ as in equation~\eqref{eq:dtilde} then Lemma~\ref{lem:adjacent-distances} holds unchanged. Theorem~\ref{thm:greedy-routing} holds as well: the only significant change to the regions shown in Figure~\ref{fig:coins} occurs if $v$ is nearly antipodal to $x_0=u$, but this only helps us since the greedy routing algorithm can get closer to $v$ by moving in any direction on its first step. Thus we can compute hybrid distance estimates $\dhat(u,v)$ as before, with error $\epshat$ as defined in~\eqref{eq:epshat}. 

If anything, the triangulation of Lemma~\ref{lem:xyz} is easier on the sphere than on the plane. We can easily find three points $x,y,z$ whose pairwise geodesic distances are within $O(\sqrt{\log n})$ of $(\pi/2)R$, so that the vectors $\vec{x},\vec{y},\vec{z}$ are nearly perpendicular. We can start with any point $x$ (conveniently, all vertices are deep) and finding $y$ and $z$ takes $O(n)$ time each: they exist with high probability since, as in flat space, a spherical cap with radius $\Omega(\sqrt{\log n})$ is nonempty with high probability.

We start our reconstruction by assuming that $\vec{x},\vec{y},\vec{z}$  are the three orthonormal basis vectors in $\R^3$. Then we can estimate $\vec{u}$ as
\[
\vec{u} 
\approx \sum_{v \in \{x,y,z\}} \vec{v} (\vec{u} \cdot \vec{v}) 
= \sum_{v \in \{x,y,z\}} \vec{v} \cos \frac{\|u-v\|}{R} 
\approx \sum_{v \in \{x,y,z\}} \vec{v} \cos \frac{\dhat(u,v)}{R} \, , 
\]
normalizing if we like so that $\|\vec{u}\|=1$. This gives $\vec{u}$ with error $O(\epshat/R)$, and mapping to the sphere of radius $R$ reconstructs $u$'s position relative to $x,y,z$ with error $O(\epshat)$. The right spherical triangle formed by our assumed positions for $x,y,z$ is congruent to their true positions up to $O(\sqrt{\log n})$. So, up to some rotation or reflection of the sphere, this yields a reconstruction with distortion $O(\epshat+\sqrt{\log n})=O(\epshat)$. 
\end{proof}

\begin{remark}
We have taken advantage of the fact that the 2-sphere has a convenient embedding in $\R^3$. A more general approach, which we claim applies to any compact Riemannian submanifold with bounded curvature, would be to work entirely within the manifold itself, building a sufficient mesh of landmarks and then triangulating using geodesic distance. In particular, in the popular model of hyperbolic embeddings (e.g.~\cite{Krioukov2010,kitsak-hyper,blasius-hyper}) where the submanifold is a ball of radius $\ell$ in a negatively curved space with radius of curvature $R$, we believe similar algorithms will work as long as $\ell/R=O(1)$. We leave this for future work.
\end{remark}
\begin{figure}
\begin{center}
\includegraphics[width=8cm]{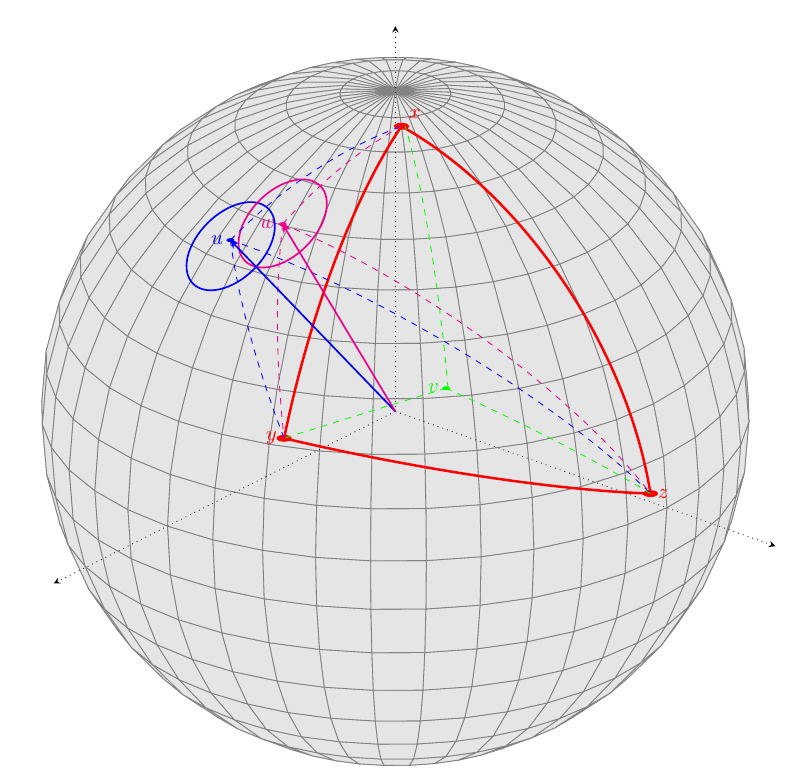}
\end{center}
\caption{Using three landmarks $x,y,z$ on the sphere to triangulate to other points in Theorem~\ref{thm:sphere}. \label{fig:3Dsphere}}
\end{figure}

\section{Missing Edges}\label{sec:missing}




In this section, we will briefly address a variant of the random geometric graph model, in which, before we are given the adjacency matrix, each edge is, independently, kept with probability $p$, and otherwise dropped. 
We will show that, in this setting, we can still accurately
reconstruct the positions of the vertices, as long as $p$ is not too small.  

There are several reasons for considering this sort of generalization of random geometric graphs.  It is a  special case of a more general model in which, after placing the vertices in the plane, each edge $(v,w)$ is included with probability $f(\|v-w\|)$ 
where $f$ is an arbitrary function, often required to be nonincreasing. For example
\newcommand{\indicator}{\mathbb{1}}
\begin{itemize}
    \item The usual RGG model is a special case where $f(x) = \indicator(x \le r)$.
    \item The model currently under discussion corresponds to $f(x) = p \;\! \indicator(x \le r)$.
    This so-called ``soft RGG'' model was studied by Penrose in~\cite{penrose-soft}.
    \item Kleinberg's~\cite{kleinberg-sw} small-world graph corresponds to $f(x) = \max\{1,x^{-C}\}$, where $C$ is a constant, although in his model only the edges are random; the vertices form an $n \times n$ square array.   He also uses $\ell_1$ distance rather than $\ell_2$.  An earlier small-world graph model due to Watts and Strogatz~\cite{watts-strogatz} corresponds to $f(x) = \indicator(x \le r) + p \;\! \indicator(x > r)$, although in that case, the underlying graph is an $n$-cycle.
\end{itemize}

In the soft RGG model, since the average degree is now $p \pi r^2$ except for vertices close to the boundary of the square, and we clearly cannot reconstruct the location of an isolated vertex,
we can only hope to reconstruct locations  when $p r^2 = \Omega(\log n)$. 
Let us make the somewhat stronger assumption that
\begin{equation} \label{eq:p^2r^2}
p^2 r^2 = \Omega(\log n).
\end{equation}
This will ensure that whenever $\|u-v\| \le r$ they almost surely have a common neighbor, so $d_G(u,v) \le 2$. 
Moreover, this condition will be sufficient to allow our short-range distance reconstruction procedure to succeed, with slight modifications. 
Under these assumptions, our main arguments still work, with only minor modifications.  The main technical issue is that, whereas before,
for a region $A \subset B(v,r) \cap B(w,r)$, we could confidently predict that 
\[
|N(v) \cap N(w) \cap A| \approx |N(v) \cap A| \approx |V \cap A| \approx \mbox{area}(A),
\]
after a $1-p$ fraction of the edges have been dropped, we expect instead that
\begin{equation}\label{eq:fn-areas}
|N(v) \cap N(w) \cap A| \approx p |N(v) \cap A| \approx p^2 |V \cap A| \approx p^2 \mbox{area}(A).
\end{equation}
By our assumption, \eqref{eq:p^2r^2}, the above expectation is large enough to ensure that as long as the area of $A$ is not too much smaller than $\pi r^2$,
$|N(v) \cap N(w) \cap A|$ is concentrated near its expectation.  For the larger quantity $|N(v) \cap A|$, the expected size is much larger, and, broadly speaking,
we can handle much smaller regions, $A$. 

\begin{proposition} \label{prop:param-estimation-false-neg}
Given the adjacency matrix of a graph $G = \RGp$, 
where $p$ and $r$ are not given, there is an efficient algorithm which estimates $p$ and $r$
accurately with error probability $O(n^{-C})$.
\end{proposition}

\begin{proof}[Sketch.] We can use the fact that for a deep vertex $v$, the expected degree is $p \pi r^2$.
And for a pair of deep vertices whose distance is $x \le r$, their expected co-degree, $|N(v) \cap N(w)|$
is $p^2 F(x)$, where $F$ is the area of $B(v,r) \cap B(w,r)$, for which a formula was given in the proof of Lemma~\ref{lem:adjacent-distances}.
This can be used to deduce that the average co-degree of two adjacent vertices is, with high probability, $C p^2 \pi r^2$, where $C$ is a 
known constant.  Combining this with the above prediction for the average degree, we can recover accurate estimates for $p$ and $r$.
\end{proof}


\noindent
{\bf Short range distance estimates}\\
Let $G = \RGp$, where  $r\, p \ge 100 \sqrt{\log n}$. 
Suppose $v,w$ are vertices for which $d_G(v,w) \le 2$.  We want to accurately estimate $\|v-w\|$ as a function of $|N(v) \cap N(w)|$. In light of ~\eqref{eq:fn-areas}, we can set
\[
\dtilde(v,w) = F^{-1}\left( \max \left\{F(r), \pi r^2 - \frac{|N(v) \cap N(w)|}{p^2} \right\}\right)
\]
where $F$, defined just before the statement of Lemma~\ref{lem:concave-MVT}, 
is the function determining the area of the lune $B(v,r) \setminus B(w,r)$, 
so that $\pi r^2 - F$ is the area of the corresponding lens $B(v,r) \cap B(w,r)$.

\begin{theorem} \label{thm:short-dis-false-neg}
Let $G = \RGp$  where  $r\, p \ge 100 \sqrt{\log n}$. 
With probability at least $1 - 2/n^2$ we have, for all vertices $v \ne w$ such that 
$d_G(v,w) \le 2$ and $v$ is deep, 
\begin{equation}
\label{eq:dtilde-good-combined-fneg}
\left| \Vert  v - w \Vert  - \dtilde(v,w) \right| \le  C \eta (\Vert v-w\Vert ) \log n \, ,
\end{equation}
where $C>0$ is a constant and $\eta: [0,2r] \to [0,1]$ is defined by 
\[
\eta(x) = \begin{cases}
\frac{1}{p} \left(\frac{2r-x}{r}\right)^{1/4} & \mbox{for $0 < x \le 2r - \frac{(\log n)^{2/3}}{r^{1/3}}$}, \\
\frac{(\log n)^{1/6}}{p \; r^{1/3}} & \mbox{for $2r - \frac{(\log n)^{2/3}}{r^{1/3}} \le x \le 2r$}.
\end{cases}
\]
In particular, $\eta(x) \le \tfrac{2^{1/4}}{p} < \tfrac{1.2}{p}$ for all $0 \le x \le 2r$.
\end{theorem}

Note that, because the asymmetric difference sets, $N(v) \setminus N(w)$ are no longer contained within the lune $B(v,r) \setminus B(w,r)$,
we are now basing our distance estimations only on the intersections $N(v) \cap N(w)$, which are still contained within the lens
$B(v,r) \cap B(w,r)$.  Consequently, our error estimate $\eta(x)$ no longer approaches zero as $x \to 0$.

\vspace{.4cm}


\noindent
{\bf Long range distance estimates:}\\
Here is our main theorem about long-range distance estimates with missing edges.

\begin{theorem}
\label{thm:greedy-routing-false-neg}
Let $G = \RGp$, where  $r\, p \ge 100 \sqrt{\log n}$. There exist absolute constants $C_1, C_2, C_3$ such that, for all $n \ge 1$, all $r \ge C_1 \sqrt{\log n}$, with probability at least $1 - C_2/n^2$, all pairs of vertices $u,v$ satisfy
\begin{equation}
\left\lceil \frac{\| u-v \|}{r} \right\rceil
\le d_G(u,v) 
\le \max \left\{ 2, \left\lceil \frac{\|u-v\| + \kappa p^{-2/3} + \kappa_2 p^{-4/3}}{r}  \right\rceil \, \right\}
\label{eq:lenses-fneg}
\end{equation}
where $\kappa$ is the same as in Theorem~\ref{thm:greedy-routing}, namely,
\begin{equation}
\label{eq:kappa-misisng-edges}
\kappa = C_3 \left( \frac{\|u-v\|}{r^{4/3}} + \frac{\log(n)}{r^{1/3}} \right) \,
\end{equation}
and $\kappa_2$ is given by
\[
\kappa_2 = \frac{C_4 (\log n)^{2/3}}{r^{1/3}}.
\]
\end{theorem}

There are three differences between the conclusion of Theorem~\ref{thm:greedy-routing-false-neg} and the original
Theorem~\ref{thm:greedy-routing}.  Firstly, the original error term, $\kappa$, is now scaled by a factor of $p^{-2/3}$,
because the greedy algorithm now makes less progress per step on average, due to the culled edges.  
There is also a second error term, $\kappa_2$ due to the possibility that the final edge for the greedy algorithm
may have been dropped.  To avoid this, we stop running the greedy algorithm once the distance to the target vertex gets
a little smaller than $2r$, and directly verify that almost surely only two more hops are required.  Finally, in case
we start out at Euclidean distance less than $r$, we still may need two hops, hence the ``max'' in the formula.


\vspace{.4cm}

\noindent
{\bf Hybrid distance estimates:}\\
Finally we use the framework of Section~\ref{sec:hybrid} to create hybrid estimates of long-range distances. That is, we set
\begin{equation}\label{eq:dhat-fneg}
    \dhat(u,v) = \min_{w: d_G(w,v) \le 2} d_1(u,w)+d_2(w,v)
\end{equation}
where $d_1$ is the long range estimate and $d_2$ is the short-range estimate defined above. Then, our main result is to 
compute the region of $r$ and $p$ for which reconstruction is possible. Specifically,
\begin{theorem}
\label{thm:dhatmind1d2-fneg}
Let $r=n^\alpha$ and $p = n^{-\gamma}$ for constants
$0 \le \gamma \le \alpha \le 1/2$.
For all pairs $u, v$ where $v$ is deep, define $\dhat(u,v)$ as in equation~\eqref{eq:dhat}. 
Then with high probability, $\dhat$ is an upper bound on the Euclidean distance $\Vert u-v\Vert $ with error 
\begin{equation}
\label{eq:epshat-fneg}
\epshat(u,v) \le C n^{\lambda/6} \log(n), \mbox{ where } \lambda = 
\max \left\{ 
6 \gamma, \  3 - 8 \alpha + 4 \gamma 
\right\}
\end{equation}
and $C$ is a constant.
\end{theorem}

For this result, we use our previously established hybrid distance construction, based on 
Theorems~\ref{thm:short-dis-false-neg} and~\ref{thm:greedy-routing-false-neg}.  
For the short-distance estimates, we have error uniformly upper-bounded by $1/p$, which is $O(r \sqrt{\log n})$.
For the long-range estimate, the error is smallest when the Euclidean distance is close to a multiple of $r$,
in which case it is approximately $\frac{\kappa}{p^{2/3}} + \frac{\kappa_2}{ p^{4/3}}$.  So the hybrid
distance is likely to select a waypoint close to a multiple of $r$ units from the landmark, and have total error
\[
\frac{1.2}{p} + \frac{\kappa}{p^{2/3}} + \frac{\kappa_2}{p^{4/3}}.
\]
Plugging in $r = n^{\alpha}$ and $p = n^{-\gamma}$, and doing a case analysis to identify the largest of the above three error terms,
we deduce the result.  

\begin{corollary}
For $p = \Theta(1)$ and $r \ge n^{3/8}$, $\dhat$ is an upper bound on the Euclidean distance $\Vert u-v\Vert $
with error $O(\log n)$.  For $r = n^{\alpha}$ and $p = n^{-\gamma}$, where $3 + 4\gamma \le 14 \alpha$
and $\gamma \le \alpha \le 1/2$, $\dhat$ is an upper bound on $\Vert u-v\Vert$ with error $O(r \log n)$.
\end{corollary}

\section{Conclusions and further work}\label{sec:conclusion}
We have shown how a combination of geometric and analytic ideas can be used to reconstruct random geometric graphs with lower distortion than in previous work~\cite{diaz2019learning}, achieving a distortion of $o(r)$ whenever $r = n^\alpha$ for $\alpha > 3/14$. Here we pose several questions for further work.

First, let us call a reconstruction $\phi$ \emph{consistent} if its distances are consistent with the graph: that is, if $(u,v) \in E$ if and only if $\Vert \phi(u)-\phi(v) \Vert \le r$. Even if $\phi$ has small distortion $d^*$, it might not be consistent: some edges $(u,v) \in E$ might have $\Vert \phi(u) - \phi(v) \Vert$ between $r$ and $r+2 d^*$, and similarly some non-neighboring pairs might have $\Vert \phi(u) - \phi(v) \Vert$ between $r-2 d^*$ and $r$. To the best of our knowledge, even finding a single consistent embedding for random geometric graphs is an open question. It might be possible to refine our embedding to make it consistent, by using ``forces'' to move neighbors slightly closer together, and push non-neighbors farther away.

Secondly, a natural question is whether we can prove a significant lower bound on the distortion. An information-theoretic approach to this question would be to show that even the Bayesian algorithm, which chooses from the uniform measure on all consistent embeddings, has a typical distortion. We have been unable to prove this. However, here we sketch an argument that there \emph{exist} consistent embeddings with a certain distortion by applying a continuous function $f$ to the square $[0,\sqrt{n}]^2$ that ``warps'' the true embedding. If $f$'s derivatives are at most $\delta$ in absolute value, then for each $v$, points close to the edge of $v$'s neighborhood may move $O(\delta r)$ closer or farther away. However, a typical $v$ has some $\eps=O(1/r)$ for which there are no points whose distance is between $r-\eps$ and $r+\eps$, since the area of the corresponding annulus is $O(1)$. This suggests that if $\delta=O(\eps/r)=O(1/r^2)$, the warped embedding is still consistent (except for a few vertices where we need to be more careful). On other other hand, even if $f$ does not change the distance between nearby vertices very much, it can still move some vertices $\delta \sqrt{n}$ from their true positions, giving a distortion $d^* = \Omega(\sqrt{n}/r^2)$. If $r=n^\alpha$ this gives $\Omega(n^{1/2-2 \alpha})$.

Even if this lower bound can be made rigorous, and even if it applies to typical consistent embeddings rather than just a few, there is a large gap between it and our upper bounds. Thus it is tempting to think that our algorithm can be improved, reducing the distortion still further.  One approach would be to try to extend the geometry of overlapping disks in Theorem~\ref{thm:short-dis} to larger graph distances. Another would be to combine them with the spectral ideas of e.g.\ \cite{araya2019}. 

Thirdly, we would like to extend the results of Section~\ref{sec:missing} to other soft RGG models where $(u,v) \in E$ with probability $f(\|u-v\|)$ for some known function $f$. It is  fairly easy to generalize our arguments to families of functions where $f(x)=0$ for $x>r$ so that the bound $d_G(u,v) \ge \|u-v\|/r$ still holds, although things change if $f(x)$ tends to zero as $x$ approaches $r$ from below. For instance, if $f(x)=(1-x/r)^\zeta$ for some constant $\zeta > 0$, the expected number of edges from $u$ to a lens of width $\delta$ and area $A$ just inside $B(u,r)$ is $\Theta((\delta/r)^\zeta A) = \Theta(\delta^{3/2+\zeta} \,r^{1/2-\zeta})$. In order for there to be $\Theta(1)$ such edges, so that the greedy routing of Theorem~\ref{thm:greedy-routing} can succeed, we need to use lenses of width $\delta = \Theta(r^{-(1-2\zeta)/(3+2\zeta)})$ rather than $\Theta(r^{-1/3})$, altering the exponent in one of our error terms. Similarly, while we can carry out short-range estimates by replacing $F(x)$ in~\eqref{eq:F} with an appropriate integral, the error bounds in  Theorem~\ref{thm:short-dis} would change when $\|u-v\|$ is close to zero or close to $2r$.

Perhaps the more interesting question is how to reconstruct vertex positions when $f(x) > 0$ for all $x \in \R^+$: for instance, if $f(x) = \e^{-x/r}$. In that case, there are a small number of ``long-range'' edges $(u,v)$ in the graph where $\|u-v\| \gg r$. As a consequence, for some pairs we have $d_G(u,v) \ll \|u-v\|/r$. Our intuition is that in this case, a greedy routing will not always give a good upper bound on $d_G(u,v)$, since it may be worth going ``out of your way'' to use these long-range edges in your path. At the same time, spectral methods may continue to perform well in this context. We leave this as a challenge for the future.

Finally, we would like to see how far these techniques can be extended to curved manifolds and submanifolds with boundary. In Theorem~\ref{thm:sphere} we took advantage of the fact that the 2-sphere has a convenient embedding in $\R^3$. A more general approach, which we claim applies to any compact Riemannian submanifold with bounded curvature, would be to work entirely within the manifold itself, building a sufficiently dense mesh of landmarks and then triangulating within mesh cells. In particular, in the popular model of hyperbolic embeddings (e.g.~\cite{Krioukov2010,kitsak-hyper,blasius-hyper}) where the submanifold is a ball of radius $\ell$ in a negatively curved space with radius of curvature $R$, we believe similar algorithms will work as long as $\ell/R=O(1)$. We leave this for future work. 
\medskip


\bibliography{references}{}

\begin{thebibliography}{10}

\bibitem{abbe-survey}
Emmanuel Abbe.
\newblock Community detection and stochastic block models: Recent developments.
\newblock {\em Journal of Machine Learning Research}, 18(177):1--86, 2018.
%
\bibitem{AbrahamCKS15}
Ittai Abraham, Shiri Chechik, David Kempe, and Aleksandrs Slivkins.
\newblock Low-distortion inference of latent similarities from a multiplex
  social network.
\newblock {\em {SIAM} J. Comput.}, 44(3):617--668, 2015.
%
\bibitem{AlmanW21}
Josh Alman and Virginia~Vassilevska Williams.
\newblock A refined laser method and faster matrix multiplication.
\newblock In {\em Proceedings of the 2021 {ACM-SIAM} {SODA}}, pages 522--539,
  2021.
%
\bibitem{araya2019}
Ernesto Araya~Valdivia and De~Castro Yohann.
\newblock Latent distance estimation for random geometric graphs.
\newblock In H.~Wallach, H.~Larochelle, A.~Beygelzimer, F.~d\textquotesingle
  Alch\'{e}-Buc, E.~Fox, and R.~Garnett, editors, {\em Advances in Neural
  Information Processing Systems}, volume~32, 2019.

\bibitem{AriasCastro}
Ery Arias-Castro, Antoine Channarond, Bruno Pelletier, and Nicolas Verzelen.
\newblock {On the estimation of latent distances using graph distances}.
\newblock {\em Electronic Journal of Statistics}, 15(1):722--747, 2021.
%
\bibitem{Aspnes04}
James~Aspnes, David K. Goldenberg, and Yang R. Yang.
\newblock On the computational complexity of sensor network localization.
\newblock In {\em Proc. ALGOSENSORS-04}, pages 32--44, 2004.
%
\bibitem{blasius-hyper}
Thomas Bl\"{a}sius, Tobias Friedrich, Anton Krohmer, and Sören Laue.
\newblock Efficient embedding of scale-free graphs in the hyperbolic plane.
\newblock {\em IEEE/ACM Transactions on Networking}, 26(2):920--933, 2018.
%
\bibitem{Bradonjic10}
Milan~Bradonji\'{c}, Robert~Elsässer, Tobias~Friedrich, Thomas~Sauerwald, and Alexandre~Stauffer.
\newblock Efficient broadcast on random geometric graphs.
\newblock In {\em 21st. ACM-SIAM SODA}, pages 1412--1421, 2010.
%
\bibitem{Breu98}
Heinz~Breu and David G. Kirkpatrick.
\newblock Unit disk graph recognition in {NP}-hard.
\newblock {\em Compt. Geometry Theory Appl.}, 9(1-2):3--24, 1998.
%
\bibitem{Bruck05}
Jehoshua~Bruck, Jie~Gao, and Anxiao~Jiang.
\newblock On the computational complexity of sensor network localization.
\newblock In {\em ACM MOBIHOC-05}, pages 181--192, 2005.
%
\bibitem{Bubeck16}
Sebastien~Bubeck, Jian~Ding, Ronen~Eldan, and Miklós~Rácz.
\newblock Testing for high-dimensional geometry in random graphs.
\newblock {\em Random Structures and Algorithms}, 49(3):503--532, 2016.
%
\bibitem{Clark90}
Brent N.~Clark, Charles J.~Colbourn, David S.~Johnson.
\newblock Unit Disk Graphs.
\newblock {\em Discrete Mathematics}, 86:165--177, 1990.
%
\bibitem{diaz2019learning}
Josep Díaz, Colin McDiarmid, and Dieter Mitsche.
\newblock Learning random points from geometric graphs or orderings.
\newblock {\em Random Structures \& Algorithms}, 2019.
%
\bibitem{diaz2016relation}
Josep Díaz, Dieter Mitsche, Guillem Perarnau, and Xavier P{\'e}rez-Gim{\'e}nez.
\newblock On the relation between graph distance and {E}uclidean distance in
  random geometric graphs.
\newblock {\em Advances in Applied Probability}, 48(3):848--864, 2016.
%
\bibitem{Dijkstra59}
Edsger W.~Dijkstra.
\newblock A Note on Two Problems in Connexion with Graphs.
\newblock {\em Numerische Mathematik}, 1:269--271, 1959.
%
\bibitem{Ellis07}
Robert~Ellis, Jeremy L. Martin, and Catherine~Yan.
\newblock Random geometric graph diameter in the unit ball.
\newblock {\em Algorithmica}, 47(4):421--438, 2007.
%
\bibitem{Gilbert}
Edward.~E. Gilbert.
\newblock Random plane networks.
\newblock {\em J. Soc. Industrial Applied Mathematics}, 9(5):533--543, 1961.
%
\bibitem{Handcock02latentspace}
Peter~D. Hoff, Adrian~E. Raftery, and Mark~S. Handcock.
\newblock Latent space approaches to social network analysis.
\newblock {\em Journal of the American Statistical Association}, 97:1090+,
  December 2002.
%
\bibitem{kitsak-hyper}
Maksim Kitsak, Ivan Voitalov, and Dmitri Krioukov.
\newblock Link prediction with hyperbolic geometry.
\newblock {\em Phys. Rev. Research}, 2:043113, 2020.
%
\bibitem{kleinberg-sw}
John Kleinberg
\newblock The Small-World Phenomenon: An Algorithmic Perspective.
\newblock In {\em Proc. of the 32nd ACM STOC}, pages 163--170, 2000.
%
\bibitem{Krioukov2010}
Dmitri Krioukov, Fragkiskos Papadopoulos, Maksim Kitsak, Amin Vahdat, and
  Mari\'an Bogu\~n\'a.
\newblock Hyperbolic geometry of complex networks.
\newblock {\em Phys. Rev. E}, 82:036106, Sep 2010.
%
\bibitem{lubold2021identifying}
Shane Lubold, Arun~G. Chandrasekhar, and Tyler~H. McCormick.
\newblock Identifying the latent space geometry of network models through
  analysis of curvature, 
\newblock {\em ArXiv:1711.02123v2}, 2021

\bibitem{moore-review}
Cristopher Moore.
\newblock The computer science and physics of community detection: Landscapes,
  phase transitions, and hardness.
\newblock {\em Bull. {EATCS}}, 121, 2017.

\bibitem{Mossel12}
Elchanan Mossel, Joe Neeman, and Allan Sly.
\newblock Reconstruction and estimation in the planted partition model.
\newblock {\em Probability Theory and Related Fields}, 162(3-4):431--461, 2015.

\bibitem{Muthu05}
Muthu~Muthukrishnan and Gopal~Pandurangan.
\newblock The bin-covering technique for thresholding random geometric graph
  properties.
\newblock In {\em 16th. ACM-SIAM SODA}, pages 989--998, 2005.

\bibitem{Parthasa17}
S.~Parthasarathy, D.~Sivakoff, M.~Tian, and Y.~Wang.
\newblock A quest to unravel the metric structure behind perturbed networks.
\newblock In {\em Proc of the 33rd SoCG}, pages 53:1--53:16, 2017.

\bibitem{MathewBook}
Mathew Penrose.
\newblock {\em Random Geometric Graphs}.
\newblock Oxford University Press, 2003.


\bibitem{penrose-soft}
Mathew Penrose.
\newblock Connectivity of soft random geometric graphs.
\newblock {\em  Annals Appl. Probab.} 26 , 986--1028, (2016). 

\bibitem{sarkar2011theoretical}
Purnamrita Sarkar, Deepayan Chakrabarti, and Andrew~W Moore.
\newblock Theoretical justification of popular link prediction heuristics.
\newblock In {\em Twenty-Second International Joint Conference on Artificial
  Intelligence}, IJCAI'11, pages 2722–--2727, 2011.


\bibitem{shalizi2019consistency}
Cosma~Rohilla Shalizi and Dena Asta.
\newblock Consistency of maximum likelihood for continuous-space network
  models
\newblock {\em ArXiv:1711.02123v2},  2019.

\bibitem{Walters}
Mark~Walters.
\newblock Random geometric graphs.
\newblock {\em Surveys in Combinatorics}, pages 365--402, 2011.

\bibitem{watts-strogatz}
Duncan J. Watts and Steven H. Strogatz.
\newblock Collective dynamics of ‘small-world’ networks.
\newblock nature {\bf 393.6684}, pages 440-442, 1998.


\bibitem{zhang2021}
Yi-Jiao Zhang, Kai-Cheng Yang, and F.~Radicchi.
\newblock Systematic comparison of graph embedding methods in practical tasks.
\newblock {\em Phys. Rev. E} 104, 044315,  2021.
\end{thebibliography}
\bibliographystyle{plain}

\end{document}